\documentclass[envcountsame]{llncs}

\usepackage{fixltx2e}

\usepackage{silence}
\usepackage{amssymb,amsmath,amsthm}
\usepackage{cite}																		  
\usepackage{graphicx}
\usepackage[colorlinks=true]{hyperref}
\usepackage{mathtools}																
\usepackage{soul}																			
\usepackage{url}
\usepackage[dvipsnames]{xcolor}
\usepackage{xifthen}
\usepackage{xparse}
\usepackage{lmodern}																	
\usepackage[T1]{fontenc}                              
\usepackage{xstring}                                  
\usepackage{paralist}                                 
\usepackage{enumitem}                                 
\usepackage{stmaryrd}																	
\usepackage{floatflt}                                 
\usepackage{wrapfig}
\usepackage{float}                                    
\usepackage[caption=false]{subfig}
\usepackage{mfirstuc}                                 

\SetSymbolFont{stmry}{bold}{U}{stmry}{m}{n}

\usepackage{tikz}
\usetikzlibrary{positioning}
\usetikzlibrary{shadows.blur}
\usetikzlibrary{shapes}
\usetikzlibrary{arrows}
\usetikzlibrary{fit}

\pgfdeclarelayer{bg}    
\pgfsetlayers{bg,main}  

\hypersetup{
  pdftitle={Revisiting MITL to Fix Decision Procedures},
  pdfauthor={Nima Roohi, Mahesh Viswanathan},
  colorlinks=true,
  citecolor={RoyalPurple},
  linkcolor = {violet},
  bookmarksopen=false,
  bookmarksnumbered=true
}

\definecolor{redish}      {rgb}{0.8, 0.1, 0.1}
\definecolor{blueish}     {rgb}{0  , 0  , 0.4}
\definecolor{greenish}    {rgb}{0  , 0.6, 0  }
\definecolor{yellowish}   {rgb}{0.8, 0.5, 0  }
\definecolor{redishBG}    {rgb}{1  , 0.3, 0.3}
\definecolor{blueishBG}   {rgb}{0.4, 0.4, 1  }
\definecolor{greenishBG}  {rgb}{0.2, 1  , 0.2}
\definecolor{yellowishBG} {rgb}{1  , 0.7, 0  }

\newtheoremstyle{black-th}{}{}{}{}{\color{black}\bfseries}{.}{ }{} 
\newtheoremstyle{red-th}{}{}{}{}{\color{redish}\bfseries}{.}{ }{} 
\newtheoremstyle{blue-th}{}{}{}{}{\color{blueish}\bfseries}{.}{ }{} 
\newtheoremstyle{green-th}{}{}{}{}{\color{greenish}\bfseries}{.}{ }{} 

\newcommand{\PutPVSmark}{ERROR}
\newcommand{\PutPVSaddr}{ERROR}
\DeclareExpandableDocumentCommand{\PVSaddress}{sm}{%
  \IfBooleanTF{#1}%
  {
    \renewcommand{\PutPVSmark}{\textsuperscript{\ref{#2}}}%
    \renewcommand{\PutPVSaddr}{}%
  }%
  {
    \renewcommand{\PutPVSmark}{\footnotemark}%
    \renewcommand{\PutPVSaddr}{\footnotetext{#2}}%
  }%
}
\newtheoremstyle{pvs-th}{}{}{}{}{}{.~}{0pt}{%
  {\bfseries%
   \thmname{#1}
   \thmnumber{ #2}
   \textsubscript{~\PVS\textsuperscript{\PutPVSmark}}%
  }%
  \thmnote{ {\bf(#3)}}
  \PutPVSaddr
} 

\newcommand{\Alg}   {Algorithm}
\newcommand{\Apx}   {Appendix}
\newcommand{\Cor}   {Corollary}
\newcommand{\Def}   {Definition}
\newcommand{\Eq}    {Equation}
\newcommand{\Fig}   {Figure}
\newcommand{\Lem}   {Lemma}
\newcommand{\Prop}  {Proposition}

\newcommand{\Sec}   {Section}
\newcommand{\Tbl}   {Table}
\newcommand{\Thm}   {Theorem}
\newcommand{\Cnd}   {Condition}
\newcommand{\Prb}   {Problem}
\newcommand{\Ex}    {Example}
\newcommand{\Itm}   {Item}
\newcommand{\Typ}   {Type}

\NewDocumentCommand{\aref}{m}{%
  \IfBeginWith{#1}{sec:} {\Sec~\ref{#1}}{%
  \IfBeginWith{#1}{thm:} {\Thm~\ref{#1}}{%
  \IfBeginWith{#1}{lem:} {\Lem~\ref{#1}}{%
  \IfBeginWith{#1}{cor:} {\Cor~\ref{#1}}{%
  \IfBeginWith{#1}{def:} {\Def~\ref{#1}}{%
  \IfBeginWith{#1}{apx:} {\Apx~\ref{#1}}{%
  \IfBeginWith{#1}{alg:} {\Alg~\ref{#1}}{%
  \IfBeginWith{#1}{fig:} {\Fig~\ref{#1}}{%
  \IfBeginWith{#1}{tbl:} {\Tbl~\ref{#1}}{%
  \IfBeginWith{#1}{rem:} {\Rem~\ref{#1}}{%
  \IfBeginWith{#1}{cnd:} {\Cnd~\ref{#1}}{%
  \IfBeginWith{#1}{prb:} {\Prb~\ref{#1}}{%
  \IfBeginWith{#1}{itm:} {\Itm~\ref{#1}}{%
  \IfBeginWith{#1}{eq:}  {\Eq~\ref{#1}}{%
  \IfBeginWith{#1}{ex:}  {\Ex~\ref{#1}}{%
  \IfBeginWith{#1}{typ:} {\Typ~\ref{#1}}{%
  \IfBeginWith{#1}{prop:}{\Prop~\ref{#1}}{%
  \errmessage{class of label '#1' is not defined.}%
	}}}}}}}}}}}}}}}}}%
}

\theoremstyle{pvs-th}
\newtheorem{pvstheorem}     [theorem] {\Thm}
\newtheorem{pvslemma}       [theorem] {\Lem}

\newtheorem{pvsproposition} [theorem] {\Prop}

\newcommand{\ProofSketch}{Proof (sketch)}
\newcommand{\ProofIdea}{Proof (idea)}

\theoremstyle{black-th}

\makeatletter
\newtheorem*{rep@theorem}{\rep@title}
\newcommand{\newreptheorem}[2]{%
\newenvironment{rep#1}[1]{%
 \def\rep@title{#2 \ref{##1}}%
 \begin{rep@theorem}}%
 {\end{rep@theorem}}}
\makeatother

\newreptheorem{corollary}   {\Corl}                         
\newreptheorem{lemma}       {\Lem}
\newreptheorem{theorem}     {\Thm}
\newreptheorem{proposition} {\Prop}

\newenvironment{myitems}{\begin{compactitem}}{\end{compactitem}}
\newenvironment{myenums}{\begin{compactenum}}{\end{compactenum}}

\DeclarePairedDelimiter{\Ceil}    {\lceil}{\rceil}

\DeclarePairedDelimiter{\Size}    {\vert}{\vert}
\DeclarePairedDelimiter{\Paren}   {\lparen}{\rparen}
\DeclarePairedDelimiter{\Brace}   {\lbrace}{\rbrace}
\DeclarePairedDelimiter{\Braket}  {\lbrack}{\rbrack}
\DeclarePairedDelimiter{\Sem}     {\llbracket}{\rrbracket}
\DeclarePairedDelimiter{\Norm}    {\lVert}{\rVert}

\DeclarePairedDelimiter{\LclRop}  {\lbrack}{\rparen}
\DeclarePairedDelimiter{\LopRcl}  {\lparen}{\rbrack}
\DeclarePairedDelimiter{\Closed}  {\lbrack}{\rbrack}
\DeclarePairedDelimiter{\Opened}  {\lparen}{\rparen}

\DeclareExpandableDocumentCommand{\LeftQ} {}{\scalebox{-0.7}[0.7]{\raisebox{-0.5mm}{?}}\!}
\DeclareExpandableDocumentCommand{\RightQ}{}{\!\scalebox{0.7}[0.7]{\raisebox{-0.5mm}{?}}}

\DeclarePairedDelimiter{\IsLClosed} {{}_{\LeftQ}\lbrack}{\rangle}
\DeclarePairedDelimiter{\IsRClosed} {\langle}{\rbrack_{\RightQ}}

\DeclarePairedDelimiter{\IsLOpen}   {{}_{\LeftQ}\lparen}{\rangle}

\DeclarePairedDelimiter{\CloseL}  {\lbrack}{\rangle}
\DeclarePairedDelimiter{\CloseR}  {\langle}{\rbrack}
\DeclarePairedDelimiter{\Open}    {\lparen\!\lvert}{\rvert\!\rparen}

\DeclarePairedDelimiter{\OpenLCloseR}{\lparen\!\lvert}{\rbrack}
\DeclarePairedDelimiter{\OpenRCloseL}{\lbrack}{\rvert\!\rparen}

\NewDocumentCommand{\GEz}     {}{{\scalebox{0.5}{${\geq}0$}}}
\NewDocumentCommand{\GTz}     {}{{\scalebox{0.6}{$+$}}}

\NewDocumentCommand{\Nat}     {}{\ensuremath{\mathbb{N}}}

\NewDocumentCommand{\Real}    {}{\ensuremath{\mathbb{R}}}
\NewDocumentCommand{\pNat}    {}{\ensuremath{\Nat_\GTz}}

\NewDocumentCommand{\pReal}   {}{\ensuremath{\Real_\GTz}}

\NewDocumentCommand{\nnReal}  {}{\ensuremath{\Real_\GEz}}

\NewDocumentCommand{\iFF} {}{\mbox{iff}}    
\NewDocumentCommand{\IFF} {}{\mbox{Iff}}    

\NewDocumentCommand{\ie}  {}{{\em i.e.}}

\NewDocumentCommand{\wLOG}{}{{\em wlog.}}

\DeclareExpandableDocumentCommand{\BNF}   {}{\texttt{BNF}}
\DeclareExpandableDocumentCommand{\NNF}   {}{\texttt{NNF}}
\DeclareExpandableDocumentCommand{\CNF}   {}{\texttt{CNF}}
\DeclareExpandableDocumentCommand{\DNF}   {}{\texttt{DNF}}
\DeclareExpandableDocumentCommand{\MTL}   {}{\texttt{MTL}}
\DeclareExpandableDocumentCommand{\LTL}   {}{\texttt{LTL}}
\DeclareExpandableDocumentCommand{\MITL}  {}{\texttt{MITL}}
\DeclareExpandableDocumentCommand{\MITLzi}{}{\MITL\textsubscript{$0,\!\infty$}}
\DeclareExpandableDocumentCommand{\MITLwi}{}{\MITL\textsubscript{WI}}
 
\NewDocumentCommand{\PSPACE}          {}{\mbox{\scalebox{0.85}{\textsf{PSPACE}}}}

\NewDocumentCommand{\EXPSPACE}        {}{\mbox{\scalebox{0.85}{\textsf{EXPSPACE}}}}

\NewDocumentCommand{\PSPACEhard}      {}{\PSPACE-hard}

\NewDocumentCommand{\PSPACEcomp}      {}{\PSPACE-complete}

\NewDocumentCommand{\EXPSPACEcomp}    {}{\EXPSPACE-complete}

\NewDocumentCommand{\PVS} {}{\ensuremath{\mathtt{PVS}}}

\NewDocumentCommand{\PVSLabel}{smm}{%
\scriptsize%
\texttt{#3}@\texttt{#2}\IfBooleanF{#1}{.}%
}

\NewDocumentCommand{\WeHave}  {}{\scalebox{0.3}{\ }\raisebox{-1ex}{\scalebox{2.3}{$\mathrel{\cdot}$}}}
\NewDocumentCommand{\SuchThat}{}{\WeHave}
\NewDocumentCommand{\filter}  {}{\ | \ }
\NewDocumentCommand{\oftype}  {}{\ensuremath{\mathrel{:}}}

\NewDocumentCommand{\goesto}  {O{}} {\ensuremath{\xrightarrow{#1}}}
\NewDocumentCommand{\Implies} {O{}} {\ensuremath{\xrightarrow{#1}}}
\NewDocumentCommand{\pset}    {m}   {2^{#1}}   
 
\NewDocumentCommand{\AAA}{}{\ensuremath{\mathcal{A}}}
\NewDocumentCommand{\BBB}{}{\ensuremath{\mathcal{B}}}
\NewDocumentCommand{\CCC}{}{\ensuremath{\mathcal{C}}}
\NewDocumentCommand{\TTT}{}{\ensuremath{\mathcal{T}}}

\NewDocumentCommand{\defEQ}     {}   {\ensuremath{\coloneqq}}
\NewDocumentCommand{\Sat}       {}   {\ensuremath{\models}}
\NewDocumentCommand{\nSat}      {}   {\ensuremath{\not\models}}
\NewDocumentCommand{\Always}    {}   {\ensuremath{\square}}
\NewDocumentCommand{\Eventually}{}   {\ensuremath{\lozenge}}

\definecolor{MyGray}{rgb}{0.45,0.45,0.45}
\DeclareExpandableDocumentCommand{\GrayColor}{}{MyGray}

\DeclareExpandableDocumentCommand{\OLDrawstr}{}{\mathtt{OLD}}
\DeclareExpandableDocumentCommand{\NEWrawstr}{}{\mathtt{NEW}}

\DeclareExpandableDocumentCommand{\OLDstr}{}{\scalebox{0.4}{\ensuremath{\color{\GrayColor}\OLDrawstr}}}
\DeclareExpandableDocumentCommand{\NEWstr}{}{\scalebox{0.4}{\ensuremath{\color{\GrayColor}\NEWrawstr}}}
\newlength{\SatLength}
\newlength{\SatHeight}
\NewDocumentCommand{\VersionedSat}{mm}{%
  \ensuremath{%
  \settowidth{\SatLength}{$\Sat$}%
  \settoheight{\SatHeight}{$\Sat$}%
  \mathbin{#1\raisebox{-0.32\SatHeight}{\hspace{-1.05\SatLength}#2}}%
  }%
}

\NewDocumentCommand{\OSat}  {}{\VersionedSat{\Sat}{\OLDstr}}
\NewDocumentCommand{\ONSat} {}{\VersionedSat{\nSat}{\OLDstr}}
\NewDocumentCommand{\NSat}  {}{\VersionedSat{\Sat}{\NEWstr}}
\NewDocumentCommand{\NNSat} {}{\VersionedSat{\nSat}{\NEWstr}}

\NewDocumentCommand{\AP}      {}{\ensuremath{\mathtt{AP}}}

\NewDocumentCommand{\allII} {}{\ensuremath{\mathcal{I}}}
\NewDocumentCommand{\nnII}  {}{\ensuremath{\mathcal{I}_\GEz}}
\NewDocumentCommand{\IIc}   {}{\ensuremath{\mathcal{I}}}
\NewDocumentCommand{\II}    {}{\ensuremath{\mathtt{I}}}
\NewDocumentCommand{\JJ}    {}{\ensuremath{\mathtt{J}}}
\NewDocumentCommand{\PP}    {}{\ensuremath{\mathcal{P}}}

\NewDocumentCommand{\Sgn}  {}{\ensuremath{f}}
\NewDocumentCommand{\Sub}  {}{\ensuremath{\mathcal{S}}}

\makeatletter
\newcommand{\Minos}{\mathbin{\text{\@Minos}}}
\newcommand{\@Minos}{%
  \ooalign{\hidewidth\raise1ex\hbox{.}\hidewidth\cr$\m@th-$\cr}%
}
\makeatother

\DeclareExpandableDocumentCommand{\UntilOp}  {}{\ensuremath{\mathcal{U}}}
\DeclareExpandableDocumentCommand{\ReleaseOp}{}{\ensuremath{\mathcal{R}}}
\NewDocumentCommand{\Until}  {mmo}{\ensuremath{#1\UntilOp\IfValueT{#3}{_{#3}}#2}}
\NewDocumentCommand{\Release}{mmo}{\ensuremath{#1\ReleaseOp\IfValueT{#3}{_{#3}}#2}}

\NewDocumentCommand{\infOBall}{}{\ensuremath{\mathtt{B}_\infty}}

\NewDocumentCommand{\iInf}{m}{\ensuremath{\underline{#1}}}
\NewDocumentCommand{\iSup}{m}{\ensuremath{\overline{#1}}}
\NewDocumentCommand{\Closure}{m}{\ensuremath{\mathtt{cl}\Paren{#1}}}

\NewDocumentCommand{\ltlG}{}{\ensuremath{\square}}
\NewDocumentCommand{\ltlF}{}{\ensuremath{\lozenge}}
\NewDocumentCommand{\ltlN}{}{\ensuremath{\scalebox{0.8}{$\bigcirc$}}}

\NewDocumentCommand{\Uwitness} {}{\ensuremath{\mathtt{witness}_{\scalebox{0.6}{\UntilOp}}}}
\NewDocumentCommand{\Rwitness} {}{\ensuremath{\mathtt{witness}_{\scalebox{0.6}{\ReleaseOp}}}}
\NewDocumentCommand{\Uproofset}{}{\ensuremath{\mathtt{proofset}_{\scalebox{0.6}{\UntilOp}}}}
\NewDocumentCommand{\Rproofset}{}{\ensuremath{\mathtt{proofset}_{\scalebox{0.6}{\ReleaseOp}}}}

\NewDocumentCommand{\UnaryFunc}{mmo}{\ensuremath{%
  \IfValueTF{#3}%
  {\IfBooleanTF{#1}%
    {#2\Paren*{#3}}%
    {#2\Paren{#3}}%
  }%
  {#2}%
}}

\NewDocumentCommand{\fvarR}{so}{\UnaryFunc{#1}{\mathtt{fvar_R}}[#2]}
\NewDocumentCommand{\fvarL}{so}{\UnaryFunc{#1}{\mathtt{fvar_L}}[#2]}
\NewDocumentCommand{\fvar}{so}{\UnaryFunc{#1}{\mathtt{fvar}}[#2]}

\NewDocumentCommand{\nnf}   {so}{\UnaryFunc{#1}{\mathtt{nnf}}[#2]}
\NewDocumentCommand{\toOld} {so}{\UnaryFunc{#1}{\mathtt{old}}[#2]}
\NewDocumentCommand{\Depth} {so}{\UnaryFunc{#1}{\mathtt{depth}}[#2]}
\NewDocumentCommand{\Height}{so}{\UnaryFunc{#1}{\mathtt{height}}[#2]}
\NewDocumentCommand{\Step}  {so}{\UnaryFunc{#1}{\mathtt{sz}}[#2]}
\NewDocumentCommand{\DS}    {so}{\UnaryFunc{#1}{\mathtt{DS}}[#2]}
\NewDocumentCommand{\Ops}   {so}{\UnaryFunc{#1}{\mathtt{ops}}[#2]}
\NewDocumentCommand{\BigO}  {so}{\UnaryFunc{#1}{\mathcal{O}}[#2]}

\NewDocumentCommand{\Term}{mmm}{%
  \expandafter\DeclareExpandableDocumentCommand\expandafter{\csname #1ls\endcsname}{}{\MakeLowercase{#2}}%
  \expandafter\DeclareExpandableDocumentCommand\expandafter{\csname #1lp\endcsname}{}{\MakeLowercase{#3}}%
  \expandafter\DeclareExpandableDocumentCommand\expandafter{\csname #1cs\endcsname}{}{#2}%
  \expandafter\DeclareExpandableDocumentCommand\expandafter{\csname #1cp\endcsname}{}{#3}%
  \expandafter\DeclareExpandableDocumentCommand\expandafter{\csname #1us\endcsname}{}{\makefirstuc{\MakeLowercase{#2}}}%
  \expandafter\DeclareExpandableDocumentCommand\expandafter{\csname #1up\endcsname}{}{\makefirstuc{\MakeLowercase{#3}}}%
} 

\NewDocumentCommand{\STORMED}{}{\mbox{\scalebox{0.9}{STORMED}}}
\NewDocumentCommand{\Buchi}{}{B\"uchi}

\Term{BA}       {\Buchi\ Automaton}                      {\Buchi\ Automata}
\Term{TS}       {Transition System}                      {Transition Systems}
\Term{HS}       {Hybrid System}													 {Hybrid Systems}
\Term{HA}       {Hybrid Automaton}                       {Hybrid Automata}
\Term{TA}       {Timed Automaton}                        {Timed Automata}
\Term{rTA}      {Rational Timed Automaton}               {Rational Timed Automata}
\Term{irTA}     {Irrational Timed Automaton}             {Irrational Timed Automata}
\Term{PolyHA}   {Polyhedral Automaton}                   {Polyhedral Automata}
\Term{nlfPolyHA}{Non-linear Polyhedral Automaton}        {Non-linear Polyhedral Automata}
\Term{RHA}      {Rectangular Automaton}                  {Rectangular Automata}
\Term{MRHA}     {Monotonic Rectangular Automaton}        {Monotonic Rectangular Automata}
\Term{iMRHA}    {Initialzied Monotonic Rectangular Automaton}        {Initialized Monotonic Rectangular Automata}
\Term{LIHA}     {Linear Inclusion  Automaton}            {Linear Inclusion Automata}
\Term{iRHA}     {Initialized Rectangular Automaton}      {Initialized Rectangular Automata}
\Term{iLIHA}    {Initialized Linear Inclusion Automaton} {Initialized Linear Inclusion Automata}
\Term{AHA}      {Affine Hybrid Automaton}                {Affine Hybrid Automata}
\Term{LEHA}     {Linear Equation Automaton}              {Linear Equation Automata}
\Term{iLEHA}    {Initialized Linear Equation Automaton}  {Initialized Linear Equation Automata}
\Term{SWHA}     {Stopwatch Automaton}                    {Stopwatch Automata}
\Term{iSWHA}    {Initialized Stopwatch Automaton}        {Initialized Stopwatch Automata}
\Term{SHA}      {Solvable Hybrid Automaton}              {Solvable Hybrid Automata}
\Term{SAOS}     {Semi-Algebraic O-Minimal System}        {Semi-Algebraic O-Minimal Systems}
\Term{SASS}			{Semi-Algebraic \STORMED\ System}			   {Semi-Algebraic \STORMED\ Systems}
\Term{tCM}			{$2$-Counter Machine}									   {$2$-Counter Machines}

\title{Revisiting MITL to Fix Decision Procedures}
  
\author{Nima Roohi\inst{1} \and Mahesh Viswanathan\inst{2}}
 
\institute{
	Department of Computer and Information Science\\
  University of Pennsylvania\footnote{%
    Part of this work was carried out while the first author was at the University of Illinois, Urbana-Champaign.}\\
  \email{roohi2@cis.upenn.edu}
  \and
  Department of Computer Science\\
  University of Illinois Urbana-Champaign\\
	\email{vmahesh@illinois.edu}}
 
\begin{document}
\maketitle  

\setlength\abovedisplayskip{0.2\baselineskip}
\setlength\belowdisplayskip{0.2\baselineskip}
\setlength\abovedisplayshortskip{0.1\baselineskip}
\setlength\belowdisplayshortskip{0.1\baselineskip}
\setlength\abovecaptionskip{\baselineskip}
\setlength\belowcaptionskip{\baselineskip}
\setlength\textfloatsep{0pt}

\begin{abstract}
\emph{Metric Interval Temporal Logic (\MITL)} is a well studied
real-time, temporal logic that has decidable satisfiability and model
checking problems. The decision procedures for {\MITL} rely on the
automata theoretic approach, where logic formulas are translated into
equivalent {\TAlp}. Since {\TAlp} are not closed under
complementation, decision procedures for {\MITL} first convert a
formula into negated normal form before translating to a {\TAls}. We
show that, unfortunately, these 20-year-old procedures are incorrect,
because they rely on an incorrect semantics of the $\ReleaseOp$
operator. We present the right semantics of $\ReleaseOp$ and give new,
correct decision procedures for {\MITL}. We show that both
satisfiability and model checking for {\MITL} are {\EXPSPACEcomp}, as
was previously claimed. We also identify a fragment of {\MITL} that we
call {\MITLwi} that is richer than {\MITLzi}, for which we show that
both satisfiability and model checking are {\PSPACEcomp}. Many of our
results have been formally proved in {\PVS}.
\end{abstract}

\section{Introduction}\label{sec:intro}



Specifications for real time systems often impose quantitative timing
constraints between events that are temporally ordered. Classical
temporal logics such as Linear Temporal Logic (\LTL)~\cite{77-LTL} are
therefore not adequate. Among the many \emph{real-time} extensions of
\LTL, the most well studied is \emph{Metric Temporal Logic
  (\MTL)}~\cite{90-MTL}. The temporal modalities in this logic, like
$\UntilOp_{\II}$, are constrained by a time interval $\II$ which
requires that the second argument of the $\UntilOp$ operator be
satisfied in the interval $\II$.
For example, the \MTL\ formula in \aref{eq:smpl-mtl} specifies that,
at all times, every request should be followed by a response within
$5$ units of time, or in case there is no response during that time,
an error should be raised within the next $0.1$ units of time.
\begin{equation}\label{eq:smpl-mtl}
\newcommand{\REQ}{\mathtt{req}}
\newcommand{\RES}{\mathtt{resp}}
\newcommand{\OUT}{\mathtt{error}}
\Always\ \ 
  \REQ\rightarrow\Paren*{
    \Eventually_{\LclRop{0,5}}\RES\vee
    \Paren*{\Always_{\LclRop{0,5}}\neg\RES\ \ \wedge\ \ \Eventually_{\LopRcl{5,5.1}}\OUT}
  }
\end{equation}

Classical decision problems for any logic are satisfiability and model
checking. For \MTL, these problems are highly undecidable; both
problems are $\Sigma^1_1$-complete~\cite{96-MITL}. Undecidability in
these cases arises because of specifications that require events to
happen exactly at certain time points, which can be described in the
logic by having temporal operators decorated by singleton intervals
(i.e., intervals containing exactly one point). If one considers the
sublogic of \MTL, called {\em Metric Interval Temporal Logic (\MITL)},
which prohibits the use of singleton intervals, then both these
problems are claimed to be \EXPSPACEcomp~\cite{96-MITL}.

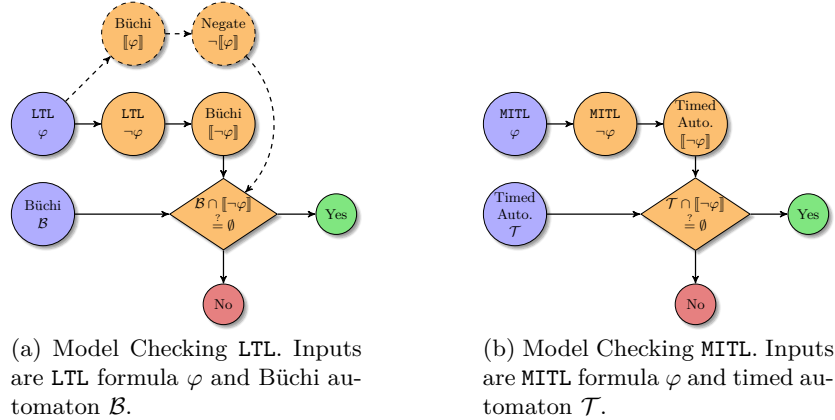
\begin{figure}[H]
\vspace{-2mm}
	\centering
  \subfloat[Model Checking \LTL.
            Inputs are \LTL\ formula $\varphi$ and \BAls\ $\BBB$.
            \label{fig:mc-ltl}]{%
    \scalebox{0.6}{\begin{tikzpicture}[>=stealth']
  \definecolor{purple}{rgb}{0.69,0.68,0.996}
  \definecolor{orange}{rgb}{0.99,0.75,0.443}
  \definecolor{green} {rgb}{0.50,0.90,0.500}
  \definecolor{red}   {rgb}{0.90,0.50,0.500}

  \tikzset{Step/.style= {circle,inner sep=0em,align=center,fill=purple,draw=black,line width=0.2mm,blur shadow={shadow blur steps=5}}}
  \tikzset{Step2/.style={circle,inner sep=0em,align=center,fill=purple,draw=black,line width=0.2mm,blur shadow={shadow blur steps=5,dashed}}}
  \tikzset{Decision/.style={diamond,aspect=1.5,inner sep=0em,align=center,fill=orange,draw=black,line width=0.2mm,blur shadow={shadow blur steps=5}}}

  \path
  (0.0,4.0) node[Step,minimum size=14mm,fill=purple] (fml) {\LTL\\$\varphi$}
  (2.0,4.0) node[Step,minimum size=14mm,fill=orange] (neg) {\LTL\\$\neg\varphi$}
  (4.0,4.0) node[Step,minimum size=14mm,fill=orange] (sem) {\Buchi\\$\Sem{\neg\varphi}$}
  (0.0,2.0) node[Step,minimum size=14mm,fill=purple] (aut) {\Buchi\\$\BBB$}
  (4.0,2.0) node[Decision]                           (cnd) {$\BBB\cap\Sem{\neg\varphi}$\\$\stackrel{?}{=}\emptyset$}
  (6.5,2.0) node[Step,minimum size=9mm,fill=green]   (yes) {Yes}
  (4.0,0.0) node[Step,minimum size=9mm,fill=red]     (no)  {No}
  (2.0,6.0) node[Step2,minimum size=14mm,fill=orange] (sem2) {\Buchi\\$\Sem{\varphi}$}
  (4.0,6.0) node[Step2,minimum size=14mm,fill=orange] (neg2) {Negate\\$\neg\Sem{\varphi}$};
  \draw[->,black,thick] (fml) -- (neg);
  \draw[->,black,thick] (neg) -- (sem);
  \draw[->,black,thick] (sem) -- (cnd);
  \draw[->,black,thick] (aut) -- (cnd);
  \draw[->,black,thick] (cnd) -- (yes);
  \draw[->,black,thick] (cnd) -- (no);
  \path[->,black,thick,dashed]
  (fml)  edge                               (sem2)
  (sem2) edge                               (neg2)
  (neg2) edge[bend left=45] node [right] {} (cnd);
\end{tikzpicture}  }}
	\hspace{15mm}
  \subfloat[Model Checking \MITL.
            Inputs are \MITL\ formula $\varphi$ and \TAls\ $\TTT$.
            \label{fig:mc-mitl}]{%
    \scalebox{0.6}{\begin{tikzpicture}[>=stealth']
  \definecolor{purple}{rgb}{0.69,0.68,0.996}
  \definecolor{orange}{rgb}{0.99,0.75,0.443}
  \definecolor{green} {rgb}{0.50,0.90,0.500}
  \definecolor{red}   {rgb}{0.90,0.50,0.500}

  \tikzset{Step/.style= {circle,inner sep=0em,align=center,fill=purple,draw=black,line width=0.2mm,blur shadow={shadow blur steps=5}}}
  \tikzset{Step2/.style={circle,inner sep=0em,align=center,fill=purple,draw=black,line width=0.2mm,blur shadow={shadow blur steps=5,dashed}}}
  \tikzset{Decision/.style={diamond,aspect=1.5,inner sep=0em,align=center,fill=orange,draw=black,line width=0.2mm,blur shadow={shadow blur steps=5}}}

  \path
  (0.0,4.0) node[Step,minimum size=14mm,fill=purple] (fml) {\MITL\\$\varphi$}
  (2.0,4.0) node[Step,minimum size=14mm,fill=orange] (neg) {\MITL\\$\neg\varphi$}
  (4.0,4.0) node[Step,minimum size=14mm,fill=orange] (sem) {Timed\\Auto.\\$\Sem{\neg\varphi}$}
  (0.0,2.0) node[Step,minimum size=14mm,fill=purple] (aut) {Timed\\Auto.\\$\TTT$}
  (4.0,2.0) node[Decision]                           (cnd) {$\TTT\cap\Sem{\neg\varphi}$\\$\stackrel{?}{=}\emptyset$}
  (6.5,2.0) node[Step,minimum size=9mm,fill=green]   (yes) {Yes}
  (4.0,0.0) node[Step,minimum size=9mm,fill=red]     (no)  {No};
  \draw[->,black,thick] (fml) -- (neg);
  \draw[->,black,thick] (neg) -- (sem);
  \draw[->,black,thick] (sem) -- (cnd);
  \draw[->,black,thick] (aut) -- (cnd);
  \draw[->,black,thick] (cnd) -- (yes);
  \draw[->,black,thick] (cnd) -- (no);
\end{tikzpicture}  }}
\caption{
  Model Checking Steps for \LTL\ and \MITL\ Formulas.
}\label{fig:mc}
\vspace{-7mm}
\end{figure}

The decision procedures for satisfiability and model checking of
{\MITL}, follow the automata theoretic approach.
In
the automata theoretic approach to satisfiability or model checking,
logical specifications are translated into automata such that the
language recognized by the automaton is exactly the set of models of
the specification. In case of {\LTL}, this involves translating formulas to
{\BAlp}, and the model checking procedure is shown in \aref{fig:mc-ltl}. 
For {\MITL}, formulas are translated into {\TAlp}. 
Model checking {\TAlp} against {\MITL} specifications is
     schematically shown in \aref{fig:mc-mitl}. The specification
     $\varphi$ is negated, a {\TAls} $\Sem{\neg\varphi}$ for
     $\neg\varphi$ is constructed, and then one checks that the system
     represented as a {\TAls} ${\cal T}$ has an empty intersection
     with $\Sem{\neg\varphi}$.
%
%
%
Since {\TAlp} are not closed under
complementation~\cite{94-TA,96-omega}, {\MITL} decision procedures
crucially rely on transforming specifications $\varphi$ (for the
satisfaction problem) and $\neg\varphi$ (for the model checking
problem) into negated normal form, \ie, one where all the negations
have been pushed all the way inside to only apply to
propositions. Using negated normal forms requires considering formulas
with the full set of boolean operators (both $\wedge$ and $\vee$) and
temporal operators (both $\UntilOp$ and its dual $\ReleaseOp$).

Unfortunately, the well known decision procedures for
{\MITL}~\cite{96-MITL} are \emph{incorrect}. This is because we show
that the semantics used for the $\ReleaseOp$ operator, which is lifted
from the semantics of $\ReleaseOp$ in {\LTL}, is not the dual of
$\UntilOp$ (see \Ex~\ref{ex:2vars}). Therefore, the {\TAlp}
constructed for the negated normal form of a formula in {\MITL}, is
not logically equivalent to the original formula.
We give a new, correct semantics for $\ReleaseOp$
(\aref{def:mtl-newsem} in \aref{sec:release}). Defining the semantics
for $\ReleaseOp$ in {\MTL} is non-trivial because of the subtle interplay
of open and closed intervals. Our definition uses 3 quantified
variables (unlike 2, which is used in the semantics of $\UntilOp$ in
both {\LTL} and {\MTL}, and $\ReleaseOp$ in {\LTL} and the incorrect
definition for {\MTL}). We show that under fairly general syntactic
conditions, one cannot use 2 quantified variables to correctly define
the semantics of $\ReleaseOp$ in {\MTL}.

We present a new translation of {\MITL} formulas into {\TAlp} that
uses our new semantics (\aref{sec:alg}). We show, using our new
construction, that the complexity of the satisfiability and model
checking problems remains {\EXPSPACEcomp} as was previous
claimed~\cite{96-MITL}. {\MITLzi} is a fragment of {\MITL} that has
{\PSPACE} decision procedures for satisfiability and model
checking. Our last result (\aref{sec:MITLwi}) shows that {\MITLzi} can
be generalized. We introduce a new, richer fragment of {\MITL} that we
call {\MITLwi}, for which we prove that satisfiability and model
checking are both {\PSPACEcomp}.

Proofs for results about {\MITL} are subtle due to the presence of a
continuous time domain and topological aspects like open and closed
sets. This is evidenced by the fact that the errors we have exposed in
this paper, have remained undiscovered for over 20 years, despite many
researchers working on problems related to {\MITL}. Therefore, to gain
greater confidence in the correctness of our claims, we have formally
proved many of our results in {\PVS}~\cite{92-PVS}. The {\PVS} proof
objects can be downloaded from
\url{http://uofi.box.com/v/PVSProofsOfMITL}.

\subsubsection{Related Work.}
The complexity of satisfiability and model checking of {\MITL}
formulas was first considered in~\cite{96-MITL}. We show that the
decision procedures are unfortunately flawed because of the use of an
incorrect semantics for $\ReleaseOp$. A different translation of
{\MITL} to {\TAlp} is presented in~\cite{06-MITL2TA}. However, their
setting is restricted in that all intervals are closed, and all
signals are continuous from the right. Note that \aref{ex:2vars} and
\aref{thm:2vars} in our paper, both use signals that are not
continuous from right.  Therefore, their algorithm does not fix the
problem in~\cite{96-MITL}. Papers~\cite{13-MITL2ATA,14-MITL2ATA}
present decision procedures for an event-based semantics for {\MITL}
which associates a time with every event. State-based semantics,
considered here and in~\cite{96-MITL,06-MITL2TA}, is very
different. For example, in a signal where $p$ is only true in the
interval $\Closed{0,c}$, there is no time that can be ascribed to the
event when $p$ first becomes false. A recent survey of results
concerning {\MTL} and its fragments can be found
in~\cite{08-RecentMTL}. Finally, robust model checking of
coFlat-\MTL\ formulas with respect to sensor and environmental noise,
is considered in~\cite{08-robust-coFlatMTL}.

\section{Preliminaries}\label{sec:prelim}

\paragraph{Sets and Functions. }
The set of {\em natural}, {\em positive natural}, {\em real}, {\em
  positive real}, and {\em non-negative real} numbers are respectively
denoted by \Nat, \pNat, \Real, \pReal, and \nnReal.
For a set $A$, power set of $A$ is denoted by $\pset{A}$, Cartesian
product of sets $A$ and $B$ is denoted by $A\times B$.  Cardinality of
$A$ is denoted by $\Size{A}$.  The set of functions from $A$ to $B$ is
denoted by $A\goesto B$.
For a set $A$, we denote the fact that $a$ is an element of $A$ by the
notation $a\oftype A$.
If $A$ is a subset of $\Real$ then for any $\epsilon\oftype\nnReal$,
we define
$\infOBall^\epsilon\Paren{A}\defEQ\Brace{x\oftype\Real\filter\exists
  a\oftype A\SuchThat\Size{x-a}\leq\epsilon}$ to be the
$\epsilon$-ball around $A$.
Finally, for any two numbers $a,b\oftype\Real$, we define $a\Minos b$
to be $\max\Brace{a-b,0}$.
 
\paragraph{Intervals.}           
Every {\em interval} of real numbers is specified by a constraint of
the form $a\lhd_1 x\lhd_2 b$, where
$a\oftype\Real\cup\Brace{-\infty}$, $b\oftype\Real\cup\Brace{
  \infty}$, and $\lhd_1,\lhd_2\oftype\Brace{<,\leq}$.  Also, if
$a\notin\Real$ (or $b\not\in\Real$) then we require $\lhd_1=<$ (or
$\lhd_2=<$).  We use the usual notation $\Closed{a,b}$,
$\Opened{a,b}$, $\LopRcl{a,b}$, and $\LclRop{a,b}$ to denote {\em
  closed}, {\em open}, {\em left-open}, or {\em right-open} intervals.
The set of {\em intervals} and {\em non-negative intervals} over
\Real, are denoted by \allII\ and \nnII, respectively.
For any interval $\II$, we use $\iInf{\II}$ and $\iSup{\II}$ to
respectively denote {\em infimum} and {\em supremum} of $\II$; if
$\II$ is empty, $\iInf{\II} = \infty$ and $\iSup{\II} = -\infty$.
{\em Width} of an interval, denoted by $\Norm{\II}$, is defined to be
$\iSup{\II} - \iInf{\II}$. Thus the width of the empty interval is
$-\infty$.
Finally, an interval with only one element is called a {\em
  singleton}; the width of such an interval (by the above definition)
is $0$.

For any interval $\II\oftype\allII$, we use
$\IsRClosed{\II}\defEQ\iSup{\II}\not\in\Real\vee\iSup{\II}\in\II$ to
check if $\II$ is closed from right.  Similarly, we use
$\IsLClosed{\II}$ and $\IsLOpen{\II}$ to check if $\II$ is closed/open
from left.  We use
$\Open{\II}\defEQ\II\setminus\Brace{\iInf{\II},\iSup{\II}}$ to denote
the interval which is achieved after removing infimum and supremum of
\II\ from it.  We also use the following intervals: 
$\CloseL{\II} \defEQ \II \cup \Brace{\iInf{\II}}$; 
$\OpenRCloseL{\II} \defEQ (\II \cup \Brace{\iInf{\II}}) \setminus \Brace{\iSup{\II}}$; and 
$\OpenLCloseR{\II} \defEQ (\II \cup \Brace{\iSup{\II}}) \setminus \Brace{\iInf{\II}}$.
 
\paragraph{Signal.}
Throughout this paper, $\AP$ is a non-empty set of {\em atomic
  propositions}~\footnote{In \aref{sec:2vars} and \aref{ex:fvar-nnf},
  we require $\Size{\AP}>1$.}.  {\em Signal} is any function of type
$\nnReal\goesto\pset{\AP}$.  Therefore, each signal is function that
defines the set of atomic propositions that are true at each instant
of time.
For a signal $\Sgn$ and time point $r\oftype\nnReal$, we define
$\Sgn^r\oftype\nnReal\goesto\pset{\AP}\oftype t\mapsto \Sgn(r+t)$ to
be another signal that shifts $\Sgn$ by $r$.


\subsection{Metric Temporal Logic}\label{sec:mtl}
In this section, we first define the syntax of metric temporal logic
(\MTL) and its subclasses metric interval temporal logic (\MITL), and
metric temporal logic with restricted intervals (\MITLzi).  We then
define the current semantics of these logics from the literature and
call this the {\em old} semantics.  Finally, we define the
transformation to a negated normal form (\nnf) and the finite
variability condition (\fvar) that are used in decision procedures
for these logics.

\begin{definition}[Syntax of \MTL, \MITL, and \MITLzi]\label{def:mtl-syn}
Syntax of a \MTL\ formula is defined using the following \BNF\ grammar, 
where by $p$ and $\II$, we mean an element of $\AP$ and $\nnII$.
\[
\varphi ::= 
  \top                            \filter
  \bot                            \filter
  p                             \filter
  \neg\varphi                     \filter
  \varphi\vee  \varphi            \filter
  \varphi\wedge\varphi            \filter
  \Until{\varphi}{\varphi}[\II] \filter
  \Release{\varphi}{\varphi}[\II]
\]
We assume $\neg$ has the highest precedence.  Syntax of a
\MITL\ formula is the same, except that singleton intervals cannot be
used. Finally, \MITLzi\ is the sublogic of \MITL\ where every interval
$\II$ appearing in a formula either has $\iInf{\II} = 0$ or
$\iSup{\II} = \infty$.
\end{definition}

\aref{def:mtl-sem}, gives the semantics of \MTL\ that was introduced
in~\cite{96-MITL} and is commonly used in the literature. Since
\MITL\ and \MITLzi\ are sublogics of \MTL, their semantics follow from
the semantics of \MTL. Later in \aref{ex:2vars}, we show that this is
not the right semantics because $\UntilOp$ and $\ReleaseOp$ are {\em
  not} duals of each other.  In \aref{def:mtl-newsem}, we introduce a
new semantics that fixes this problem.  We distinguish the two
semantics by putting words $\OLDrawstr$ and $\NEWrawstr$, in gray,
below the satisfaction relation $\Sat$ (\aref{def:mtl-sem} uses
\OSat\ and \aref{def:mtl-newsem} uses \NSat).

\begin{definition}[Semantics of \MTL]\label{def:mtl-sem}
Let $\Sgn\oftype\nnReal\goesto\pset{\AP}$ be an arbitrary signal.  For
a \MTL\ formula $\varphi$, satisfaction relation $\Sgn\OSat\varphi$ is
defined using the following inductive rules:
\[
\newcommand{\DEF}{\mbox{\ \iFF\ }}
\begin{array}{@{}lll@{}}
\Sgn\OSat\top                                 & & \mbox{is always true}\\
\Sgn\OSat\bot                                 & & \mbox{is always false}\\
\Sgn\OSat p                                   &\DEF&p\in\Sgn(0)\\
\Sgn\OSat\neg\varphi                          &\DEF&\neg\Paren*{\Sgn\OSat\varphi}\\
\Sgn\OSat\varphi_1\vee  \varphi_2             &\DEF&\Paren*{\Sgn\OSat\varphi_1}\vee  \Paren*{\Sgn\OSat\varphi_2}\\
\Sgn\OSat\varphi_1\wedge\varphi_2             &\DEF&\Paren*{\Sgn\OSat\varphi_1}\wedge\Paren*{\Sgn\OSat\varphi_2}\\
\Sgn\OSat\Until{\varphi_1}{\varphi_2}[\II]    &\DEF&\exists t_1\oftype\II\SuchThat\Paren*{\Sgn^{t_1}\OSat\varphi_2}\wedge
                                                    \forall t_2\oftype\Paren{0,t_1}\WeHave\Sgn^{t_2}\OSat\varphi_1\\
\Sgn\OSat\Release{\varphi_1}{\varphi_2}[\II]  &\DEF&\forall t_1\oftype\II\WeHave\Sgn^{t_1}\OSat\varphi_2\ \ \ \vee\\
                                              &    &\exists t_1\oftype\pReal\SuchThat\Paren*{\Sgn^{t_1}\OSat\varphi_1}\wedge
                                                    \forall t_2\oftype\Braket{0,t_1}\cap\II\WeHave\Sgn^{t_2}\OSat\varphi_2
\end{array}
\]
Finally, $\Sgn\ONSat\varphi$ is defined to be
$\neg\Paren*{\Sgn\OSat\varphi}$.
\end{definition}

The decision procedures for satisfiability and model checking of
\MITL\ introduced in~\cite{96-MITL}, rely on translating the formulas
into timed automata. Since timed languages are not closed under
complementation~\cite{94-TA}, complementation cannot be handled as a
first-class operation. Instead, one constructs an equivalent formula,
where the negations are pushed all the way inside to only apply to
propositions. We present this definition of the negation normal form
(\aref{def:nnf}) of a \MTL\ formula next. The implicit assumption is
that a formula is semantically equivalent to its negation normal form
for certain special signals that are said to be \emph{finitely
  variable}. We will define finite variability after presenting the
definition of negation normal form.

\begin{definition}[Negated Normal Form]\label{def:nnf}
For any \MTL\ formula $\varphi$, its negated normal form, denoted by
$\nnf[\varphi]$, is a formula that is obtained by pushing all the
negations inside operators. It is formally defined using the following
inductive rules ($p\oftype\AP$ is an atomic formula, and $\varphi_1$
and $\varphi_2$ are arbitrary \MTL\ formulas):
\[
\scalebox{0.895}{$%
\begin{array}{@{}l@{}}
\begin{array}{@{}lll@{\ \ \ \ \ \ \ }lll@{\ \ \ \ \ \ \ }lll@{\ \ \ \ \ \ \ }lll@{}}
\nnf[\top]           &\defEQ&\top &
\nnf[\neg\top]       &\defEQ&\bot &
\nnf[\phantom{\neg}p]&\defEQ&\phantom{\neg} p 
\\
\nnf[\bot]    &\defEQ&\bot  & 
\nnf[\neg\bot]&\defEQ&\top &
\nnf[\neg p]  &\defEQ&\neg p 
&
\nnf[\neg \neg \varphi]&\defEQ&\nnf[\varphi]
\end{array}\\
\begin{array}{@{}lll@{\ \ \ \ \ \ }lll@{}}
\nnf[\varphi_1\vee  \varphi_2]&\defEQ&\nnf[\varphi_1]\vee  \nnf[\varphi_2] &
\nnf[\Until  {\varphi_1}{\varphi_2}[\II]]&\defEQ&\Until  {\nnf[\varphi_1]}{\nnf[\varphi_2]}[\II] 
\\
\nnf[\varphi_1\wedge\varphi_2]&\defEQ&\nnf[\varphi_1]\wedge\nnf[\varphi_2]	&
\nnf[\Release{\varphi_1}{\varphi_2}[\II]]&\defEQ&\Release{\nnf[\varphi_1]}{\nnf[\varphi_2]}[\II]
\\
\nnf[\neg\Paren{\varphi_1\vee  \varphi_2}]&\defEQ&\nnf[\neg\varphi_1]\wedge\nnf[\neg\varphi_2] &
\nnf[\neg\Paren{\Until  {\varphi_1}{\varphi_2}[\II]}]&\defEQ&\Release{\nnf[\neg\varphi_1]}{\nnf[\neg\varphi_2]}[\II] 
\\
\nnf[\neg\Paren{\varphi_1\wedge\varphi_2}]&\defEQ&\nnf[\neg\varphi_1]\vee  \nnf[\neg\varphi_2] &
\nnf[\neg\Paren{\Release{\varphi_1}{\varphi_2}[\II]}]&\defEQ&\Until  {\nnf[\neg\varphi_1]}{\nnf[\neg\varphi_2]}[\II]
\end{array}
\end{array}$}
\]
\end{definition}

The semantics of the modal operators $\UntilOp$ and $\ReleaseOp$ are
defined using quantifiers, and both of them are $\exists\forall$
formulas. However, $\UntilOp$ and $\ReleaseOp$ are supposed to be
duals of each other (see \aref{def:nnf}) eventhough they are defined
using formulas with the same quantifier alternation. Thus, $\UntilOp$
and $\ReleaseOp$ work as duals only for special signals that are
\emph{finitely variable}~\cite{96-MITL,04-Logics4RT,08-RecentMTL}.
Intuitively, it means during any finite amount of time, number of
times a signal changes its value is finite.  \aref{def:fvar}
formalizes this condition.

\begin{definition}[Finite Variability]\label{def:fvar}
For an implicitly known satisfaction relation $\Sat$, a signal
$\Sgn\oftype\nnReal\goesto\pset{\AP}$ is said to be \emph{finitely
  variable from right with respect to a \MTL\ formula $\varphi$},
denoted by $\fvarR[\Sgn,\varphi]$, \iFF\
\[ 
\begin{array}{r@{}l}
\forall r\oftype\nnReal\WeHave& \Paren*{\forall\epsilon\oftype\pReal\WeHave  \exists t\oftype\Paren{r,r+\epsilon}\SuchThat\Sgn^t\Sat\varphi} \Rightarrow \\
                              & \Paren*{\exists\epsilon\oftype\pReal\SuchThat\forall t\oftype\Paren{r,r+\epsilon}\WeHave  \Sgn^t\Sat\varphi}
\end{array}
\] 
$\Sgn$ is said to be \emph{finitely variable from left with respect to
  a \MTL\ formula $\varphi$}, denoted by $\fvarL[\Sgn,\varphi]$, \iFF\
\[ 
\begin{array}{r@{}l}
\forall r\oftype\pReal\WeHave& \Paren*{\forall\epsilon\oftype\Paren{0,r}\WeHave  \exists t\oftype\Paren{r-\epsilon,r}\SuchThat\Sgn^t\Sat\varphi} \Rightarrow \\
                             & \Paren*{\exists\epsilon\oftype\Paren{0,r}\SuchThat\forall t\oftype\Paren{r-\epsilon,r}\WeHave  \Sgn^t\Sat\varphi}
\end{array}
\] 
$\Sgn$ is said to be \emph{finitely variable with respect to a
  \MTL\ formula $\varphi$}, denoted by $\fvar[\Sgn,\varphi]$,
\iFF\ $\fvarL[\Sgn,\varphi]\wedge\fvarR[\Sgn,\varphi]$.
$\Sgn$ is said to be \emph{finitely variable (from left/right)}
\iFF\ for any \MTL\ formula $\varphi$, $\Sgn$ is finitely variable
(from left/right) with respect to $\varphi$.
Whenever we use finite variability, precise definition of $\Sat$ will
be clear from the context.
\end{definition}

Finite variability as defined here (\aref{def:fvar}), is formulated
differently than the definition given
in~\cite{08-RecentMTL,04-Logics4RT}. However, the two definitions are
equivalent, and we prefer the presentation given here because it makes
the quantifier alternation in the definition explicit.

\aref{def:fvar} suggests that to establish finite variability of a
signal, we need to consider all possible \MTL\ formulas. However, it
is known that a signal is finitely variable \iFF\ it is finitely
variable over all atomic formulas; we will prove that this observation
also holds for the new semantics for $\ReleaseOp$ that we define in
the next section (\aref{lem:fvar}).

Every finitely variable signal can be specified using (finite or
countably infinite) sequence of intervals paired with subsets of
atomic propositions that are true during that interval. For example,
$\Paren{\Closed{0,1},\Brace{p}}, \Paren{\Opened{1,4},\Brace{q}},
\Paren{\LclRop{4,\infty},\Brace{p,q}}$ specifies a signal that is
$\Brace{p}$ during $\Closed{0,1}$, $\Brace{q}$ during $\Opened{1,4}$,
and $\Brace{p,q}$ at all other times. All our examples use this
representation for (finitely variable) signals.

Equivalence for formulas in \MTL\ will only be considered with respect
to finitely variable signals. That is, two \MTL\ formulas $\varphi_1$
and $\varphi_2$ are said to be {\em equivalent}, \iFF\ for any
finitely variable signal \Sgn\ we have
$\Paren{\Sgn\Sat\varphi_1}\Leftrightarrow\Paren{\Sgn\Sat\varphi_2}$;
here $\Sat$ can either be taken to be the relation defined in
\aref{def:mtl-sem} or the one we will define in the next section
(\aref{def:mtl-newsem}).

\section{Defining the Semantics of Release}\label{sec:release}

The semantics of release as defined in \aref{def:mtl-sem} does not
ensure that $\ReleaseOp$ and $\UntilOp$ are duals. \Ex~\ref{ex:2vars}
describes a finite variable signal $\Sgn$ such that
$\Sgn\ONSat\Until{p}{q}[\II]$ (for propositions $p,q$ and interval
$\II$) and $\Sgn\ONSat\Release{\neg p}{\neg q}[\II]$. Thus, the
transformation to negation normal form described in \aref{def:nnf}
does not preserve the semantics, making decision procedures for
satisfiability and model checking outlined in~\cite{96-MITL}
incorrect. In this section, we identify the correct semantics of the
release operator so that the transformation to negation normal form
described in \aref{def:nnf} is semantically correct. Our semantics for
$\ReleaseOp$ is more complicated than the one in \aref{def:mtl-sem},
in that uses 3 quantified variables. We conclude this section by
establishing that this increase in expression complexity is necessary
--- it is impossible to define the semantics of $\ReleaseOp$ using a
$\exists\forall$ formula that uses only two quantified variables.

\noindent
\begin{minipage}{0.65\textwidth}
\begin{example}\label{ex:2vars}
{
Let $c\oftype\Open{\II}$ be an arbitrary point, and define $\Sgn$ to be the signal 
$\Paren*{\Braket{0,c},\Brace{p}},\allowbreak\Paren*{\Paren{c,\infty},\Brace{q}}$.  
Clearly, $\Sgn\ONSat\Until{p}{q}[\II]$ and hence $\Sgn\OSat\neg\Paren*{\Until{p}{q}[\II]}$.  
On the other hand, $\neg q$ is not true throughout $\II$ and whenever $\neg p$ is true, $\neg q$ is false.  
Therefore, $\Sgn\ONSat\Release{\neg p}{\neg q}[\II]$. 
Thus, \aref{def:nnf} does not preserve 
\unskip\parfillskip 0pt \par
}
\end{example}
\end{minipage}
\hfill
\begin{minipage}{0.32\textwidth}
\vskip-1.5\baselineskip
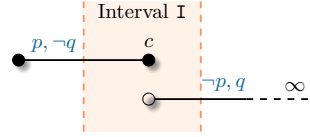
\begin{figure}[H]
\scalebox{0.865}{\begin{tikzpicture}[>=stealth']
  \definecolor{orange1}{rgb}{1.000,0.95,0.914}
  \definecolor{orange2}{rgb}{0.988,0.56,0.388}
  \definecolor{blue}   {rgb}{0.075,0.40,0.600}

  \tikzset{Inclusive/.style= {circle,fill=black,draw=black,line width=0mm,inner sep=.2em,blur shadow={shadow blur steps=5}}}
  \tikzset{Exclusive/.style= {circle,           draw=black,line width=0mm,inner sep=.2em,blur shadow={shadow blur steps=5}}}
  \tikzset{Interval/.style={inner sep=0em,align=center,fill=purple,draw=black,line width=0.2mm,blur shadow={shadow blur steps=5,dashed}}}

  \fill[fill=orange1] (1,-0.5) rectangle (2.8,1.5);
  \draw[color=orange2,dashed,thick] (1.0,-0.5) -- (1.0,1.5);
  \draw[color=orange2,dashed,thick] (2.8,-0.5) -- (2.8,1.5);
  \node at (1.9,1.35) {Interval \II};

  \path
  (0.0,0.6) node[Inclusive]           (zero) {}
  (2.0,0.6) node[Inclusive,label=$c$] (c1)   {}
  (2.0,0.0) node[Exclusive]           (c2)   {}
  (3.5,0.0) coordinate                (inf1) {}
  (4.5,0.0) coordinate                (inf2) {}
  ;

  \path[black,thick,text=blue]
  (zero) edge node [near start,above] {$p,\neg q$} (c1)
  (c2)   edge node [near end,  above] {$\neg p,q$} (inf1)
  ;

  \path[black,thick,dashed]
  (inf1) edge node [near end, above] {$\infty$} (inf2)
  ;
\end{tikzpicture}}
\vskip-\baselineskip
\caption{Signal \Sgn}\label{fig:2vars-sig}
\end{figure}
\end{minipage}
\vskip-0.3\baselineskip

\noindent
the semantics,
making decision procedures for satisfiability and model
checking~\cite{96-MITL} of \MITL\ that first convert a formula into
negation normal form, incorrect.

Since the semantics of release is incorrect (from the perspective of
ensuring that $\UntilOp$ and $\ReleaseOp$ are duals), we define a new
semantics for the release operator. Denseness of the time domain,
along with subtleties introduced to due to open and closed endpoints
of intervals, make proofs about \MTL\ challenging to get
right. Therefore, to have greater confidence in our results, we have
proved most of our results in Prototype Verification Systems
(\PVS)~\cite{92-PVS}. We explicitly mark all lemmas and theorems that
were proved in \PVS~\footnote{Each such result is annotated by
  $\langle\textsf{lemma-name}\rangle@\langle\textsf{theory-name}\rangle$.
  Theory name \textsf{thry} can be found in a file named
  \textsf{thry.pvs}.}. Space limitations prevent these formal proofs
to be part of this paper. However they can be downloaded from
\url{http://uofi.box.com/v/PVSProofsOfMITL}.

\begin{definition}[New Semantics for \MTL]\label{def:mtl-newsem}
Let $\Sgn\oftype\nnReal\goesto\pset{\AP}$ be an arbitrary signal and
$r\oftype\nnReal$ be an arbitrary point in time.  For a \MTL\ formula
$\varphi$, we define the satisfaction relation
$\Sgn\NSat\varphi$ as follows.
\[
\newcommand{\DEF}{\mbox{\ \iFF\ }}
\begin{array}{@{}lll@{}}
\Sgn\NSat\top                                 & & \mbox{is always true}\\
\Sgn\NSat\bot                                 & & \mbox{is always false}\\
\Sgn\NSat p                                   &\DEF&p\in\Sgn(0)\\
\Sgn\NSat\neg\varphi                          &\DEF&\neg\Paren*{\Sgn\NSat\varphi}\\
\Sgn\NSat\varphi_1\vee  \varphi_2             &\DEF&\Paren*{\Sgn\NSat\varphi_1}\vee  \Paren*{\Sgn\NSat\varphi_2}\\
\Sgn\NSat\varphi_1\wedge\varphi_2             &\DEF&\Paren*{\Sgn\NSat\varphi_1}\wedge\Paren*{\Sgn\NSat\varphi_2}\\
\Sgn\NSat\Until{\varphi_1}{\varphi_2}[\II]    &\DEF&\exists t_1\oftype\II\SuchThat\Paren*{\Sgn^{t_1}\NSat\varphi_2}\wedge
                                                    \forall t_2\oftype\Paren{0,t_1}\WeHave\Sgn^{t_2}\NSat\varphi_1\\
\Sgn\NSat\Release{\varphi_1}{\varphi_2}[\II]  &\DEF&\forall t_1\oftype\II\WeHave\Sgn^{t_1}\NSat\varphi_2\ \ \ \vee\\
                                              &    &\exists t_1\oftype\pReal\SuchThat\Paren*{\Sgn^{t_1}\NSat\varphi_1}\wedge
                                                    \forall t_2\oftype\Braket{0,t_1}\cap\II\WeHave\Sgn^{t_2}\NSat\varphi_2\ \ \ \vee\\
                                              &    &\exists t_1\oftype\CloseL{\II},t_2\oftype\II\cap\Paren{t_1,\infty}\SuchThat
                                                    \forall t_3\oftype\II\WeHave\Paren{t_3\leq t_1\Implies\Sgn^{t_3}\NSat\varphi_2}\\
                                              &    &\hfill\wedge\Paren{t_1 < t_3\leq t_2\Implies\Sgn^{t_3}\NSat\varphi_1}
\end{array}
\]
$\Sgn\NNSat\varphi$ is defined to be $\neg\Paren*{\Sgn\NSat\varphi}$.
\end{definition}

\begin{example}\label{ex:2vars-cont}
Consider the signal $\Sgn$ from \Ex~\ref{ex:2vars} that does not
satisfy $\Until{p}{q}[\II]$. Observe that $\Sgn\NSat\Release{\neg p}{\neg q}[\II]$ by meeting the third condition for satisfying
release operators under the new semantics as follows. Take $t_1 = c$,
and $t_2 = c+\epsilon$ such that $\Braket{c,c+\epsilon} \subseteq
\II$. Now, for any $t_3 \leq t_1$, $\Sgn^{t_3}\NSat \neg q$, and for
any $t_1 < t_3 \leq t_2$, $\Sgn^{t_3}\NSat \neg p$.
\end{example}

We will show that the new semantics (\aref{def:mtl-newsem}) ensures
that the transformation to negation normal form (\aref{def:nnf})
preserves the semantics when considering finite variable
signals. Before presenting this result (\aref{thm:duality}), we recall
that a signal is finitely variable iff the truth of \emph{every}
formula in the logic changes only finitely many times within any
bounded time. This is difficult to establish. Instead,
in~\cite{96-MITL}, it was shown that proving the finite variablity of
a signal with respect to atomic propositions, guarantees its finite
variability with respect to all formulas. We show that such an
observation is also true for the new semantics we have defined.

\PVSaddress{%
  \PVSLabel*{mtl}{atom\_finitevar\_eqv\_fml\_finitevar\_left} and
\PVSLabel{mtl}{atom\_finitevar\_eqv\_fml\_finitevar\_right}}
\newcommand{\RESfvar}{%
Using the semantics in \Def~\ref{def:mtl-newsem}, for any signal $\Sgn$, the following conditions hold:
\[
\begin{split}
  \Paren*{\forall p\oftype\AP\WeHave\fvarL[\Sgn,p]}\Leftrightarrow\Paren*{\forall \varphi\oftype\MTL\WeHave\fvarL[\Sgn,\varphi]}\\
  \Paren*{\forall p\oftype\AP\WeHave\fvarR[\Sgn,p]}\Leftrightarrow\Paren*{\forall \varphi\oftype\MTL\WeHave\fvarR[\Sgn,\varphi]}
\end{split}
\]}
\begin{pvslemma}[Finite Variability]\label{lem:fvar}
\RESfvar
\end{pvslemma}

We now present the main result about the correctness of the new
semantics.
\PVSaddress{\PVSLabel{mtl}{sat\_eqv\_nnfsat}}
\newcommand{\RESnnfEqvSat}{%
  If a signal $\Sgn$ is finitely variable from right then for any \MTL\ formula $\varphi$, $\Sgn\NSat\varphi$ \iFF\ $\Sgn\NSat\nnf[\varphi]$.}
\begin{pvstheorem}[Duality]\label{thm:duality}
  \RESnnfEqvSat
\end{pvstheorem}

We conclude this section by introducing a new (defined) temporal
operator that we will use. For any \MTL\ formula $\varphi$, let
$\ltlN\varphi$ be defined as
$\Release{\varphi}{\varphi}[\Paren{0,\infty}]$.  Intuitively,
$\Sgn\NSat\ltlN\varphi$ \iFF\ $\varphi$ becomes true and stays true
for some positive amount of time.  \Prop~\ref{prop:next} formalizes
this observation.  Note that instead if $\infty$ in definition of
$\ltlN\varphi$, one can use any other positive number and obtain an
equivalent
formula~\footnote{\label{ft:next}\PVSLabel*{mtl}{next\_def\_1},
  \PVSLabel{mtl}{next\_def\_3}}.  In writing formulas, we assume
$\ltlN$ has higher precedence than $\UntilOp$ and $\ReleaseOp$
operators but lower precedence than $\vee$ and $\wedge$ operators.

\PVSaddress*{ft:next}
\newcommand{\RESnext}{%
  Let $\Sat$ be the satisfaction relation given in \Def~\ref{def:mtl-sem} or \Def~\ref{def:mtl-newsem}.
  For any signal $\Sgn$ we have
  $\Sgn\Sat\ltlN\varphi$ \iFF\
  $\exists\epsilon\oftype\pReal\SuchThat\forall t\oftype\Paren{0,\epsilon}\WeHave\Sgn^t\Sat\varphi$.}
\begin{pvsproposition}[Operator \ltlN]\label{prop:next}
\RESnext
\end{pvsproposition}

The correctness of our semantics (\aref{thm:duality}) was only
established for signals that were finitely variable from the
right. Unfortunately, our next example shows that this assumption
cannot be relaxed. 

\begin{example}\label{ex:fvar-nnf}
Let $\varphi$ be the following formula. 
\[
\Paren*{\ltlN q} \wedge \neg\Paren*{\Until{p}{q}[\Paren{0,1}]} \wedge
\neg\Paren*{\Until{\neg p}{q}[\Paren{0,1}]}
\]
$\varphi$ is satisfied by a signal that is finitely variable from the
left as follows. Consider $\Sgn$ to be such that $q$ is true at all
times, and $p$ is true at times $t = \frac{1}{2^n}$ for $n \in \Nat$
and false at all other times. First observe that $\Sgn$ is finitely
variable from the left and $\Sgn\Sat\varphi$, no matter whether $\Sat$
is given by \aref{def:mtl-sem} or by \aref{def:mtl-newsem}.

Putting $\varphi$ into its \NNF\ we obtain the following formula which
is not satisfiable (using either \Def~\ref{def:mtl-sem} or
\Def~\ref{def:mtl-newsem}).
\[
\Paren*{\ltlN q} \wedge \Paren*{\Release{\neg p}{\neg q}[\Paren{0,1}]}
\wedge \Paren*{\Release{p}{\neg q}[\Paren{0,1}]}
\]
\end{example}

\subsection{Necessity of Using Three Variables}\label{sec:2vars}

The new semantics of the release operator, given in
\aref{def:mtl-newsem}, is defined by quantifying over 3 time points. A
natural question to ask is whether this is necessary. Is there a
``simpler'' definition of the semantics of the release operation? In
this section, we show that this is in some sense impossible. We show
that no first order definition of the semantics of release that
quantifies over only two time points can be correct.

Let us fix the formula $\varphi = \neg\Paren{\Until{p}{q}[\II]}$,
where $p$ and $q$ are proposition. The goal is to show that $\neg
\varphi$ cannot be expressed by a ``simple''
$\exists\forall$-formula. Let us first define what we mean by
``simple'' formulas. Let ${\cal L}_{p,q,\II}$ be the collection of
first order formulas of the form
\begin{equation}
\label{eq:def-form}
  \bigwedge_{i\oftype\Brace{1,\ldots,n}}\bigvee_{j\oftype\Brace{1,\ldots,i_n}}
  \exists t_1\oftype\Real\SuchThat\forall t_2\oftype\Real\WeHave\phi_{i,j}\Paren*{\Sgn,t_1,t_2}
\end{equation}
Here $n\oftype\Nat$ and $i_n\oftype\Nat$, and formula $\phi_{i,j}$ is
given by the \BNF\ grammar
\[
        \phi  ::=
        \neg\phi
           \filter
        \phi\vee\phi
           \filter
        \alpha_1 t_1+\alpha_2 t_2\bowtie\beta   \filter
        \Sgn^{t_u}\Sat\psi
\]
where $\alpha_1,\alpha_2,\beta\oftype\Real$ are arbitrary constants,
$\bowtie\oftype\Brace{<,\leq,>,\geq}$ is an arbitrary relation symbol,
$u\oftype\Brace{1,2}$, and $\psi\oftype\Brace{p,q,\neg p, \neg
  q}$. We assume $\Sat$ here is either $\OSat$ or $\NSat$; it doesn't
make a difference because $\psi$ is propositional. We call constraints
of the form $\alpha_1 t_1+\alpha_2 t_2\bowtie\beta$ \emph{domain}
constraints and constraints of the form $\Sgn^{t_u}\Sat\psi$
\emph{signal} constraints.

Before presenting the main theorem of this section, we examine the
restrictions imposed on formulas in ${\cal L}_{p,q,\II}$.  The
requirements that $\phi \in {\cal L}_{p,q,\II}$ be in conjuctive
normal form, or that there be no $\vee$ or $\wedge$ operations between
quantifiers, or that $\psi$ in the \BNF\ only be $\Brace{p,q,\neg
  p,\neg q}$ do not restrict the expressive power. Any formula not
satifying these conditions can be transformed into one that does. The
main restrictions are that all domain constraints are linear and that
$\Sgn$ in the signal constraints only be shifted by $t_1$ or $t_2$ and
not by an arithmetic combination of them.

\begin{theorem}\label{thm:2vars}
There is no formula in ${\cal L}_{p,q,\II}$ that is logically
equivalent to $\neg\Paren{\Until{p}{q}[\II]}$ over finite variable
signals. In fact, for any $\phi \in {\cal L}_{p,q,\II}$, there are
signals $\Sgn_1$ and $\Sgn_2$ in which the truth of any atomic
proposition changes at most 2 times such that
$\Sgn_1\Sat\neg\Paren{\Until{p}{q}[\II]}$,
$\Sgn_2\Sat\Paren{\Until{p}{q}[\II]}$ but either both $\Sgn_1,\Sgn_2$
satisfy $\phi$ or neither does.
\end{theorem}

The rest of the section is devoted to proving
\aref{thm:2vars}. Suppose (for contradiction)
\[
\phi = 
  \bigwedge_{i\oftype\Brace{1,\ldots,n}}\bigvee_{j\oftype\Brace{1,\ldots,i_n}}
  \exists t_1\oftype\Real\SuchThat\forall t_2\oftype\Real\WeHave\phi_{i,j}\Paren*{\Sgn,t_1,t_2}
\]
is logically equivalent to $\neg\Paren*{\Until{p}{q}[\II]}$. We begin
by observing that $\phi$ can be assumed to be in a special canonical
form. We then identify two parameters $r$ and $\delta$ that are used
in the construction of signals that demonstrate the inequivalence
of $\phi$ and $\neg\Paren*{\Until{p}{q}[\II]}$. Finally, we use these
parameters to construct the signals and prove the inequivalence.

\subsubsection{Cannonical Form of $\phi$.}
We can assume without loss of generality, that $\phi$ has the
following special form.
\begin{myenums}
\item \label{smpl:neg-free} Negations are pushed all the way inside,
  and are only applied to $p$ or $q$. This is always possible since
  $\Brace{<,\leq,>,\geq}$ is closed under negation and
  $\neg\Paren{\Sgn^t\Sat\psi}$ is, by definition, equivalent to
  $\Sgn^t\Sat\neg\psi$. Note that after this step, $\phi_{i,j}$
  may contain $\wedge$ operator.
\item \label{smpl:cnf} Each $\phi_{i,j}$ is a conjunction of clauses
  that we denote as $\phi_{i,j,k}$.
\item \label{smpl:S} Every clause in $\phi_{i,j}$, has at most one
  signal constraint of the form $\Sgn^{t_1}\Sat\psi_1$ and one signal
  constraint of the form $\Sgn^{t_2}\Sat\psi_2$ where $\psi_1$ and
  $\psi_2$ are boolean combinations of $p$ and $q$.
\item \label{smpl:P} For an arbitrary clause $\phi_{i,j,k}$ in
  $\phi_{i,j}$, let $S$ and $P$ be, respectively, the set of signal
  and \emph{negated} domain constraints in
  $\phi_{i,j,k}$. $\phi_{i,j,k}$ is equivalent to
  $\Paren*{\bigwedge_{\theta\oftype
      P}\theta}\Implies\Paren*{\bigvee_{s\oftype S}s}$. The left hand
  side of this implication defines a 2-dimensional convex polyhedron
  using variables $t_1$ and $t_2$.

  For the rest of the proof, \wLOG, we assume every clause in every
  $\phi_{i,j}$ is of the form $P\Implies S$, where $P$ is a polyhedron
  over $t_1$ and $t_2$, and $S$ is a disjunction of $0$, $1$, or $2$
  signal constraints. For any polyhedron $P$, we define
  $\Sem{P}\defEQ\Brace{\Paren{t_1,t_2}\filter P(t_1,t_2)}$ to be the
  set of points in $P$. Also, $\Closure{\Sem{P}}$ is defined to be the
  closure of $\Sem{P}$. Finally, let $\PP$ be the set of all polyhedra
  used in $\phi$.
\end{myenums}

\setlength\abovecaptionskip{0pt}%
\setlength\belowcaptionskip{0pt}%
\begin{wrapfigure}{r}{29mm}%
\raggedleft%
\vspace{-1.9\baselineskip}%
\includegraphics[width=29mm]{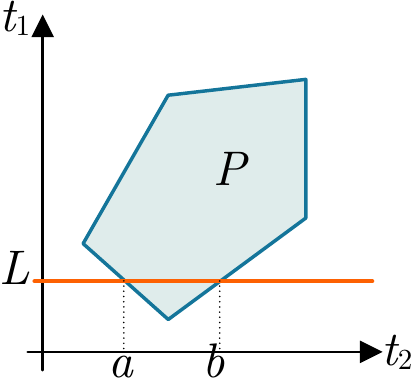}%
\caption{}\label{fig:conditional}%
\vspace{-2.0\baselineskip}%
\end{wrapfigure}%
\setlength\abovecaptionskip{\baselineskip}%
\setlength\belowcaptionskip{\baselineskip}%
\Fig~\ref{fig:conditional} shows a geometrical interpretation of the
polyhedral representation of clauses in $\phi_{i,j}$.  Let
$\phi_{i,j,k}$ be a clause that is specified by $P\Implies S$. An
arbitrary horizontal line $L$, may or may not have intersection with
$P$. Either way, $L$ {\em witnesses} $\exists t_1\SuchThat\forall
t_2\WeHave\phi_{i,j,k}$ \iFF\ for all points in this possibly empty
intersection, $S$ is satisfied. Every $\phi_{i,j}$ is a set of
constraints of the form $P\Implies S$. Therefore, $L$ {\em witnesses}
$\exists t_1\SuchThat\forall t_2\WeHave\phi_{i,j}$ \iFF\ $L$ witnesses
all clauses in $\phi_{i,j}$. Furthermore, $\exists t_1\SuchThat\forall
t_2\WeHave\phi_{i,j}$ is true \iFF\ there is a horizontal line $L$
that witnesses it.

\subsubsection{Identifying Parameters $\delta$ and $r$.}
For any $P\oftype\PP$, define $V_P$ to be the set of vertices of $P$,
and $L_P$ to be the set of points on \emph{vertical} edges of $P$
(that is, segments of a line of the form $t_2 = c$ for some $c\oftype\Real$).  Let
$C_1\defEQ\bigcup_{P\oftype\PP}\Paren{V_P\cup L_P}$. Define
$C_2\defEQ\Brace{t_2\oftype\Real\filter\exists
  t_1\oftype\Real\SuchThat\Paren{t_1,t_2}\in C_1}$ be the projection
of points in $C_1$ on $t_2$. Take $C_3\defEQ C_2$ if $\iSup{\II} =
\infty$, and $C_3 \defEQ C_2 \cup \Brace{\iSup{\II}}$, otherwise.
Observe that $C_3$ is always a finite set.
Therefore, for some $\epsilon\oftype\pReal$,
$\II\setminus\infOBall^\epsilon\Paren{C_3} \neq \emptyset$. Fix $r \oftype\II\setminus C_3$ such that for some $\epsilon$,
$\infOBall^\epsilon(r) \subseteq \II\setminus C_3$.

For any $P\oftype\PP$ and $c\oftype\Real$,
$\Closure{\Sem{P}}\cap\Sem{t_1\ =c}$ is equal to $\Brace{c}\times J$,
for some (possibly empty) interval $J$.  Define
$\Norm{\Closure{\Sem{P}}\cap\Sem{t_1=c}}$ to be $\Norm{J}$. The main
property we exploit about our choice of $r$, is that if $(r,c) \in J$
then $\Norm{J}$ is either $\leq 0$ or ``large''. This is the content
of our next lemma.
\begin{lemma}
\label{lem:rndelta}
There is a $\delta\oftype\pReal$ such that for any $P \oftype\PP$ and
$c\oftype\Real$, we have that if $\Paren{c,r}\in\Closure{\Sem{P}}$
then
$\Norm{\Closure{\Sem{P}}\cap\Sem{t_1=c}}\notin\LopRcl{0,\delta}$.
\end{lemma}
\begin{proof}
For the purpose of contradiction, let us assume that the lemma does
not hold. Since $\PP$ is a finite set, we therefore have, 
\[
\exists P\oftype\PP\SuchThat                                         
         \forall n\oftype\pNat\WeHave
         \exists c_n\oftype\Real\SuchThat
         \Paren{c_n,r}\in\Closure{\Sem{P}}\wedge\Norm{\Closure{\Sem{P}}\cap\Sem{t_1=c_n}}\in\Paren*{0,\frac{1}{n}}
\]
Let $P$ be the polyhedron witnessing the violation of the lemma as in
the above equation.  If $\Sem{P}$ is an empty set, point or a
line/line segment/half line that is not horzontal then its
intersection with $\Sem{t_1=c_n}$ is either empty or has width $0$,
which contradicts the fact that $P$ violates the lemma. Otherwise, if
$\Sem{P}$ is a horizontal line, a horizontal half line, or a
horizontal line segment, its intersection with $\Sem{t_1=c_n}$ is
either empty or has a fixed width, which again contradicts $P$
violating the lemma. Therefore, consider the case when $P$ has a
non-empty interior. Since $P$ has a finite number of vertices, an
infinite subsequence of $\Paren{c_n,r}$ converges to a vertex of $P$.
However, this is also a contradiction since our choice of $r$ ensures
that $\Paren{c_n,r}$ is always $\epsilon$ away from any point in
$C_1$.
\end{proof}
  
%
\begin{wrapfigure}{r}{30mm}%
\raggedleft
\vskip-2\baselineskip
\includegraphics[width=30mm]{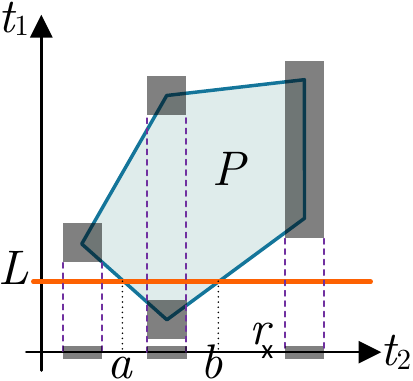}
\vskip-\baselineskip
\caption{}\label{fig:r}
\vskip-3\baselineskip
\end{wrapfigure}
For the rest of this section, let us fix $r$ as above, and take
$\delta$ to be such that in addition to \aref{lem:rndelta}, it satisfies $\iSup{\II} - r > \delta$. \Fig~\ref{fig:r}
shows a geometric interpretation for the parameters $r,\delta$ we have
identified. For any clause $P\Implies S$ in $\phi$ and for any
horizontal line $L$ defined by $t_1=c$ (for any $c\oftype\Real$), if
$\Paren{c,r}\in\Sem{P}$ then we have
\begin{inparaenum}
\item either $a = b = r$, or
\item if $S$ contains a constraint of the form $f^{t_2}\Sat\psi$ then
  then $S$ is checked for all values of $t_2$ in an inveral of size
  $>\delta$ around $r$.
\end{inparaenum}

\subsubsection{Constructing Signal $\Sgn_1$.}
\aref{fig:signals-1} shows the signal $\Sgn_1$.  $\Sgn_1$ is the
signal $\Paren{\Closed{0,r},\Brace{p,\allowbreak\neg q}},
\Paren{\LopRcl{r,r+\delta},\Brace{\neg p,q}},
\Paren{\Opened{r+\delta,\infty},\Brace{p,\neg q}}$.  It is easy to see
that $\Sgn_1\Sat\neg\Paren{\Until{p}{q}[\II]}$ (where $\Sat$ is either
$\OSat$ or $\NSat$); the reason is similar to
\aref{ex:2vars}. Therefore, if $\phi$ is equivalent to
$\neg\Paren*{\Until{p}{q}[\II]}$, then $\Sgn_1$ also satisfies $\phi$.

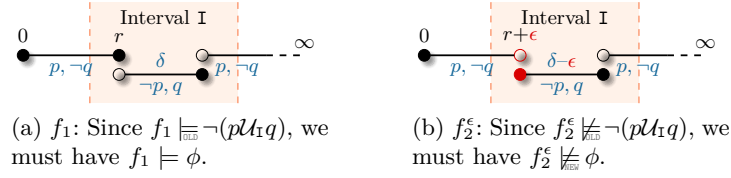
\begin{figure}[t]
\centering
  \subfloat[$\Sgn_1$:
            Since $\Sgn_1\OSat\neg\Paren{\Until{p}{q}[\II]}$, we must have $\Sgn_1\Sat\phi$.
            \label{fig:signals-1}]{\scalebox{0.85}{\begin{tikzpicture}
  \definecolor{orange1}{rgb}{1.000,0.95,0.914}
  \definecolor{orange2}{rgb}{0.988,0.56,0.388}
  \definecolor{blue}   {rgb}{0.075,0.40,0.600}

  \tikzset{Inclusive/.style= {circle,fill=black,draw=black,line width=0mm,inner sep=.2em,blur shadow={shadow blur steps=5}}}
  \tikzset{Exclusive/.style= {circle,           draw=black,line width=0mm,inner sep=.2em,blur shadow={shadow blur steps=5}}}
  \tikzset{Interval/.style={rectangle, fill=orange1}}
  \tikzset{IntervalE/.style={rectangle, draw=orange2,line width=0.2mm,dashed}}

  \path
  (0.0, 0.0) node[Inclusive,label=$0$]           (0i)    {}
  (1.5, 0.0) node[Inclusive,label=$r$]           (1i)    {}
  (1.5,-0.3) node[Exclusive          ]           (1e)    {}
  (2.8,-0.3) node[Inclusive          ]           (4i)    {}
  (2.8, 0.0) node[Exclusive          ]           (4e)    {}
  (3.8, 0.0) coordinate                          (inf1)  {}
	(4.4, 0.0) coordinate[label=$\infty$]          (inf2)  {}
	(2.2, 0.6) node	 (dummy) {Interval \II}
  ;

  \path[black,thick,text=blue]
  (0i)   edge         node [below] (l1) {$     p,\neg q$} (1i)
  (1e)   edge         node [below] (l2) {$\neg p,     q$} node [above] (l7) {$\delta$} (4i)
  (4e)   edge         node [below] (l5) {$     p,     \neg q$}	(inf1)	(inf1) edge[dashed] node [below] (l6) {}	(inf2)
  ;

  \begin{pgfonlayer}{bg} 
  \path
  node[Interval,anchor=north,inner xsep=10pt,inner ysep=0,fit=(1e)(l2)(dummy)] (interval) {}
	(interval.north west) edge[IntervalE] (interval.south west) {}
	(interval.north east) edge[IntervalE] (interval.south east) {}
  ;
  \end{pgfonlayer}

\end{tikzpicture}}}
  \hspace{10mm}
  \subfloat[$\Sgn^\epsilon_2$:
            Since $\Sgn^\epsilon_2\ONSat\neg\Paren{\Until{p}{q}[\II]}$, we must have $\Sgn^\epsilon_2\NNSat\phi$.
            \label{fig:signals-2}]{\scalebox{0.85}{\begin{tikzpicture}
  \definecolor{orange1}{rgb}{1.000,0.95,0.914}
  \definecolor{orange2}{rgb}{0.988,0.56,0.388}
  \definecolor{blue}   {rgb}{0.075,0.40,0.600}
  \definecolor{red}    {rgb}{0.850,0.00,0.000}

  \tikzset{Inclusive/.style= {circle,fill=black,draw=black,line width=0mm,inner sep=.2em,blur shadow={shadow blur steps=5}}}
  \tikzset{Exclusive/.style= {circle,           draw=black,line width=0mm,inner sep=.2em,blur shadow={shadow blur steps=5}}}
  \tikzset{Interval/.style={rectangle, fill=orange1}}
  \tikzset{IntervalE/.style={rectangle, draw=orange2,line width=0.2mm,dashed}}

	\newcommand{\EPS}{\textcolor{red}{\ensuremath{\epsilon}}}

  \path
  (0.0, 0.0) node[Inclusive,label=$0$]           (0i)    {}
  (1.5, 0.0) node[Exclusive,label=$r$+\EPS,color=red]      (1i)    {}
  (1.5,-0.3) node[Inclusive,               color=red]           (1e)    {}
  (2.8,-0.3) node[Inclusive          ]           (4i)    {}
  (2.8, 0.0) node[Exclusive          ]           (4e)    {}
  (3.8, 0.0) coordinate                          (inf1)  {}
	(4.4, 0.0) coordinate[label=$\infty$]          (inf2)  {}
	(2.2, 0.6) node	 (dummy) {Interval \II}
  ;

  \path[black,thick,text=blue]
  (0i)   edge         node [below] (l1) {$     p,\neg q$} (1i)
  (1e)   edge         node [below] (l2) {$\neg p,     q$} node [above] (l7) {$\delta$--\EPS} (4i)
  (4e)   edge         node [below] (l5) {$     p,     \neg q$}	 (inf1)
	(inf1) edge[dashed] node [below] (l6) {} (inf2)
  ;

  \begin{pgfonlayer}{bg} 
  \path
  node[Interval,anchor=north,inner xsep=10pt,inner ysep=0,fit=(1e)(l2)(dummy)] (interval) {}
	(interval.north west) edge[IntervalE] (interval.south west) {}
	(interval.north east) edge[IntervalE] (interval.south east) {}
  ;
  \end{pgfonlayer}

\end{tikzpicture}}}
\caption{Signals $\Sgn_1$ and $\Sgn^\epsilon_2$ (interval widths are not proportional).}
\label{fig:signals}
\end{figure}

\subsubsection{Constructing signal $\Sgn^\epsilon_2$.}
Since $\Sgn_1$ satisfies $\phi$, there are $c_1,c_2,\ldots
c_n\oftype\Real$ and $j_1,\ldots,j_n\oftype\Nat$ such that for any
$i\oftype\Brace{1,\ldots,n}$, line $L_i\defEQ\Paren{t_1=c_i}$
witnesses the satisfaction of $\exists t_1\SuchThat\forall
t_2\WeHave\ \phi_{i,j_i}$.
Consider a clause $\phi_{i,j_i,k}$ of $\phi_{i,j_i}$ of the form
$P_{i,j_i,k} \Implies S_{i,j_i,k}$. We know that line $L_i$ witnesses
the satisfaction of this clause. Define interval $\JJ_{i,j_i,k}$ to be
the interval given by $\Sem{P_{i,j_i,k}}\cap\Sem{L_i} = \Brace{c_i}
\times \JJ_{i,j_i,k}$. Choose $\epsilon_{i,j_i,k}\oftype\pReal$ to be
such that either $\Paren{r,r+\epsilon_{i,j_i,k}} \subseteq
\JJ_{i,j_i,k}$ or $\Paren{r,r+\epsilon_{i,j_i,k}} \cap \JJ_{i,j_i,k} =
\emptyset$. Such a choice of $\epsilon_{i,j_i,k}$ always exists no
matter what $\JJ_{i,j_i,k}$ is. Fix the parameter $\epsilon$ to be
\[
\epsilon \defEQ \frac{1}{2} \min 
     \Paren{\Brace{\epsilon_{i,j_i,k}\filter\mbox{any }i,j_i,k}\cup
            \Brace{c_i-r\filter c_i>r}}.
\]
Notice that our choice of $\epsilon$ ensures that for $i,j_i,k$,
$\Paren{r,r+\epsilon}$ is either contained in $\JJ_{i,j_i,k}$ or is
disjoint from it.

Having defined $\epsilon$, we are ready to describe the signal
$\Sgn^\epsilon_2$ which is shown in
\aref{fig:signals-2}. $\Sgn^\epsilon_2$ is given as
$\Paren{\LclRop{0,r+\epsilon},\Brace{p,\neg q}},
\Paren{\Closed{r+\epsilon,r+\delta},\Brace{\neg p,q}},
\Paren{\Opened{r+\delta,\infty},\Brace{p,\neg q}}$. Notice $\Sgn_1$
and $\Sgn^\epsilon_2$ only differ in the interval
$\Paren{r,r+\epsilon}$. Further, $f^\epsilon_2$ satisfies
$\Until{p}{q}[\II]$. 

\subsubsection{Deriving a Contradiction.}
Let $c_1,c_2,\ldots c_n$, $L_1,L_2,\ldots L_n$, and $j_1,\ldots j_n$,
as defined above, be the witness that demonstrates that $\Sgn_1$
satisfies $\phi$. We will show that these also witness the fact that
$\Sgn^\epsilon_2$ satisfies $\phi$, giving us the desired
contradiction. That is, we will show that the lines $L_i \defEQ
\Paren{t_1=c_i}$ witness the fact that $f^\epsilon_2$ satisfies
$\exists t_1\SuchThat\forall t_2\WeHave\ \phi_{i,j_i}$. Consider any
clause $P_{i,j_i,k} \Implies S_{i,j_i,k}$ of $\phi_{i,j_i}$.
\begin{myitems}
\item Suppose $S_{i,j_i,k}$ is of the form $\Sgn^{t_1}\Sat\psi$, where
  $\psi$ is a boolean combination of propositions $p,q$. Observe that
  by construction $t_1 \not\in \Paren{r,r+\epsilon}$, and so
  $\Sgn_1(t_1) = \Sgn^\epsilon_2(t_1)$. Therefore, since $\Sgn_1$
  satisfies $S_{i,j_i,k}$, so does $\Sgn^\epsilon_2$.
\item Suppose $S_{i,j_i,k}$ is of the form $\Sgn^{t_2}\Sat\psi$. Let
  $\JJ_{i,j_i,k}$ be as defined above. By our choice of $\epsilon$, we
  know that either $\Paren{r,r+\epsilon} \cap \JJ_{i,j_i,k} =
  \emptyset$ or $\Paren{r,r+\epsilon} \subseteq \JJ_{i,j_i,k}$. In the
  first case, we have $\Sgn_1(t) = \Sgn^\epsilon_2(t)$ for all $t \in
  \JJ_{i,j_i,k}$. Therefore, $\Sgn_1$ satisfies $\forall t_2 \in
  \JJ_{i,j_i,k}\SuchThat\ S_{i,j_i,k}$ iff $\Sgn^\epsilon_2$ satisfies
  the same condition. Now, let us consider the more interesting case
  when $\Paren{r,r+\epsilon} \subseteq \JJ_{i,j_i,k}$. Observe that in
  this case $r \in \Closure{\Sem{\JJ_{i,j_i,k}}}$, and so
  \aref{lem:rndelta} applies, and we can conclude that
  $\Norm{\Closure{\Sem{\JJ_{i,j_i,k}}}} > \delta$. This means that
  either there is a $t < r$ such that $t \in \JJ_{i,j_i,k}$ or there
  is a $t > r+\delta$ such that $t \in \JJ_{i,j_i,k}$. Thus, for any
  $t_2 \in \Paren{r,r+\epsilon}$, there is a $t \in \JJ_{i,j_i,k}$
  such that $\Sgn^\epsilon_2(t_2) = \Sgn_1(t)$. Hence, once again we
  can conclude that $L_i$ witnesses the satisfaction of $\forall
  t_2\SuchThat S_{i,j_i,k}$ by $\Sgn^\epsilon_1$ since $\Sgn_1$ does.
\item The last case to consider is when $S_{i,j_i,k}$ is of the form
  $\Sgn^{t_1}\Sat\psi_1\vee\Sgn^{t_2}\Sat\psi_2$. In this case also we
  can conclude that $f^\epsilon_2$ satisfies this clause using the
  reasoning in the previous two cases.
\end{myitems}

\section{Satisfiability and Model Checking \MITL\ Formulas}\label{sec:alg}

The satisfiability and model checking problems for \MITL\ are as
follows. In satisfiability, given an \MITL\ formula $\varphi$, one
needs to determine if there is a finite variable signal $\Sgn$ that
satisfies $\varphi$. The model checking problem is, given a
\TAls\ ${\cal T}$ and a \MITL\ formula $\varphi$, determine if every
finite variable signal produced by ${\cal T}$ satisfies
$\varphi$. Algorithms for both these problems rely on translating the
\MITL\ formula $\varphi$ (or its negation, in the case of model
checking) to a \TAls\ $\Sem{\varphi}$ and then either checking
emptiness of $\Sem{\varphi}$ (for satisfiability) or checking the
emptiness of the intersection of two \TAlp\ (for model
checking). Since \TAlp\ are not closed under complementation, decision
procedures rely on translating a formula in {\NNF}. As observed in
\Ex~\ref{ex:2vars}, since the semantics of $\ReleaseOp$ is incorrect,
the decision procedures for satisfiability and model checking given
in~\cite{96-MITL,04-Logics4RT,08-RecentMTL} are incorrect. In this
section, we describe a translation of \MITL\ to \TAlp\ with respect to
the correct semantics given in \aref{def:mtl-newsem}.

The translation given in~\cite{96-MITL} from \MITL\ in \NNF\ to
{\TAlp}, is correct when the semantics of $\ReleaseOp$ is taken to be
as given in \aref{def:mtl-sem}. We will exploit this construction to
give a translation with respect to the semantics in
\aref{def:mtl-newsem}. More precisely, in \aref{def:old}, we transform
an \MITL\ formula $\varphi$ into $\toOld[\varphi]$ such that for any
signal $\Sgn$, we have
$\Paren{\Sgn\NSat\varphi}\Leftrightarrow\Paren{\Sgn\OSat\toOld[\varphi]}$.

\begin{definition}\label{def:old}
The transformation $\toOld$ is inductively defined as follows. In this
definition, $\varphi'_1$ and $\varphi'_2$ are $\toOld[\varphi_1]$ and
$\toOld[\varphi_2]$, respectively.
\[
\begin{array}{@{}l@{}}
\begin{array}{@{}lll@{\ \ \ \ \ \ \ }lll@{\ \ \ \ \ \ \ }lll@{\ \ \ \ \ \ \ }lll@{}}
\toOld[\top]																&\defEQ& \top																	&	
\toOld[p]																		&\defEQ& p																		&	
\toOld[\varphi_1\vee  \varphi_2]						&\defEQ& \varphi'_1\vee  \varphi'_2						
\\
\toOld[\bot]																&\defEQ& \bot																	&	
\toOld[\neg\varphi]													&\defEQ& \neg\toOld[\varphi]									&
\toOld[\varphi_2\wedge\varphi_2]						&\defEQ& \varphi'_1\wedge\varphi'_2						
\end{array}
\\
\begin{array}{@{}l@{\ }l@{\ }l@{}}
\toOld[\Until  {\varphi_1}{\varphi_2}[\II]]	&\defEQ& \Until{\varphi'_1}{\varphi'_2}[\II]	\\
\toOld[\Release{\varphi_1}{\varphi_2}[\II]]	&\defEQ& 																			
	\begin{array}[t]{@{}l@{\ }l@{}}
		\Paren{\Release{\varphi'_1}{\varphi'_2}[\II]}\vee\Paren{\Release{\ltlN\varphi'_1}{\varphi'_2}[\II]}		&\mbox{if }\iInf{\II}>0	\\
		\Paren{\Release{\varphi'_1}{\varphi'_2}[\II]}\vee\Paren{\Release{\ltlN\varphi'_1}{\varphi'_2}[\II]}\vee\ltlN\varphi'_1 
		&\mbox{if }\iInf{\II}=0\wedge\IsLOpen{\II}	\\
		\Paren{\Release{\varphi'_1}{\varphi'_2}[\II]}\vee\Paren{\Release{\ltlN\varphi'_1}{\varphi'_2}[\II]}\vee\Paren{\varphi'_2\wedge\ltlN\varphi'_1}
		&\mbox{otherwise}\\
	\end{array}
\end{array}
\end{array}
\]
\end{definition}

The transformation $\toOld$ ensures that the semantics of the
transformed formula $\toOld[\varphi]$ with respect to $\OSat$, is the
same as the semantics of $\varphi$ with respect to $\NSat$.

\PVSaddress{\PVSLabel{mtl}{sat\_and\_isat}}
\begin{pvslemma}\label{lem:old2new}
For any signal $\Sgn$ and \MTL\ formula $\varphi$, we have $\Paren{\Sgn\NSat\varphi}\Leftrightarrow\Paren{\Sgn\OSat\toOld[\varphi]}$.
\begin{proof}[\ProofIdea]
Use induction on the structure of $\varphi$.
Special treatment for $\iInf{\II}$ is needed because both \aref{def:mtl-sem} and \aref{def:mtl-newsem} define what is called \emph{strict} 
semantics, in which value of signal at time $0$ is {\em not} important when $\iInf{\II}>0$.
\end{proof}
\end{pvslemma}

It is worth emphasizing that \aref{def:old} and \aref{lem:old2new}
apply to any \MTL\ formula (not just {\MITL}), and the soundness of
the transformation holds for \emph{any} signal (and not just finite
variable signals).

\aref{lem:old2new} immediately gives us a procedure for transforming a
negated normal formula into a \TAls\ according to
\aref{def:mtl-newsem}.  For any \MITL\ formula $\varphi$, we transform
$\toOld[\nnf[\varphi]]$ into a \TAls\ according to~\cite{96-MITL}.
Note that output of $\toOld$ is in \NNF\ \iFF\ its input
is~\footnote{\PVSLabel{mtl}{toISatNNF}}.  Using \aref{thm:duality} and
\aref{lem:old2new} we know that the transformation is correct.

The main problem with this approach and \aref{def:old}, is that
$\toOld[\varphi]$ could be exponentially larger than $\varphi$. So we
need to address the concern that this might change the complexity of
satisfiability and model checking.
The complexity of the transformation in~\cite{96-MITL} for {\MITL} and
{\MITLzi} depends only on the number of \emph{distinct} subformulas in
$\varphi$, and \emph{not} on the formula size of $\varphi$
itself~\footnote{The complexity depends on the size of the DAG
  representation of the formula, and not its syntactic
  representation.}. In \aref{prop:subsize}, we show that the number of
subformulas of $\toOld[\varphi]$ is \emph{linearly} related to the
number of subformulas of $\varphi$.
Thus using the construction in~\cite{96-MITL} for
$\toOld[\nnf[\varphi]]$ does not change the complexity results for
satisfiability and model checking~\footnote{There are multiple initial
  transformations in~\cite{96-MITL}, and each one of them can make the
  size of formula exponentially bigger. However, the number of
  distinct subformulas remains linear to the size of original
  formula.}.

\newcommand{\PROPsubsize}{%
For any \MTL\ formula $\varphi$, we have $\Size{\Sub_{\toOld[\varphi]}} \leq 6\Size{\Sub_\varphi}$, where for any \MTL\ formula $\psi$, $\Sub_\psi$ is the set 
of subformulas of $\psi$ (including $\psi$, itself).}
\begin{proposition}\label{prop:subsize}
\PROPsubsize
\end{proposition}

The proof of \aref{prop:subsize} is defered to the appendix in the
interests of space. It is worth noting that this proposition also
applies to any \MTL\ formula and not just \MITL. Using
\aref{prop:subsize}, we can conclude that the complexity of
satisfiability and model checking remain unchanged in the new
semantics.

\begin{corollary}\label{cor:complexity}
With respect to the semantics in \aref{def:mtl-newsem}, the
satisfiability and model checking problems for \MITLzi\ and
\MITL\ are \PSPACEcomp\ and \EXPSPACEcomp, respectively.
\end{corollary}

\section{\MITL\ with Wide Intervals (\MITLwi)}\label{sec:MITLwi}

One important result in~\cite{96-MITL} is the identification of
sublogic \MITLzi\ of \MITL, for which the satisfiability and model
checking problems are in \PSPACE, as opposed to \EXPSPACE\ for
\MITL. In this section we prove that this result can be generalized.
We identify a more expressive sublogic of \MITL\ for which
satisfiability and model checking are in \PSPACE.

For a formula $\varphi$ of \MITL, the size of $\varphi$ is the size of
the formula, where the constants appearing in the intervals are
represented in binary. Here we do not restrict constants in $\varphi$
to be natural numbers (as in~\cite{96-MITL}), but instead allow them
to be rational numbers; as is standard, we represent a rational number
as a pair of binary strings encoding the numerator and denominator of
the fractional representation. Define \MITLwi\ to be the collection of
\MITL\ formulas $\varphi$ such that every interval $\II$ appearing in
$\varphi$, either has 
\begin{inparaenum}
\item $\iInf{\II} = 0$ or 
\item $\iSup{\II} =\infty$ or
\item $\frac{\iSup{\II}}{\iSup{\II} - \iInf{\II}} \leq n$
when $0 < \iInf{\II} < \iSup{\II} < \infty$, where $n$ is the size of
$\varphi$. \label{cnd:WI}
\end{inparaenum}

Notice that every \MITLzi\ formula is also a \MITLwi\ formula, and
there are many \MITLwi\ formulas that are not \MITLzi\ formulas. Thus,
\MITLwi\ is a richer fragment of \MITL. Condition~\ref{cnd:WI} above in the
definition of \MITLwi\ says that when there is an interval not
conforming to the restrictions of \MITLzi, and it has a large
supremum, then the size of the interval must also be large. Thus,
intervals in \MITLwi\ can be thought of as being ``wide'' (and hence
the name). The main result of this section is the following.
\begin{theorem}\label{thm:MITLwi}
For any \MITLwi\ formula $\varphi$ of size $n$, there is a
\TAls\ $\Sem{\varphi}$ satisfying the following properties.
\begin{myenums}
\item For any finite variable signal $\Sgn$, $\Sgn$ is in the language
  of $\Sem{\varphi}$ \IFF\ $\Sgn\NSat\varphi$.
\item $\Sem{\varphi}$ has at most $2^{\BigO[n^2]}$ many locations and edges.
\item $\Sem{\varphi}$ has at most $\BigO[n^2]$ clocks.
\item $\Sem{\varphi}$ has at most $\BigO[n]$ distinct integer
  constants, each bounded by $2^{\BigO[n]}$.
\end{myenums}
Furthermore, $\Sem{\varphi}$ can be constructed in polynomial space
from $\varphi$.
\end{theorem}

The proof of this result will be presented over the course of this
section, but it is worth noting that \aref{thm:MITLwi} immediately
gives a \PSPACE\ algorithm for satisfiability and model checking of
\MITLwi.
\begin{corollary}\label{cor:MITLwi}
Model checking and satisfiability problems for \MITLwi\ is \PSPACEcomp.
\begin{proof}
Being in \PSPACE\ is an immediate consequence of \aref{thm:MITLwi},
and \PSPACEhard ness follows from the \PSPACEhard ness of \MITLzi.
\end{proof}
\end{corollary}

The rest of this section is devoted to proving \aref{thm:MITLwi}. We
begin by highlighting the crucial features of \MITLzi\ that make it
easier to decide than \MITL\ (\aref{sec:MITLvsMITLzi}). In
\aref{sec:witness}, we sketch the proof of \aref{thm:MITLwi}, by
drawing on the observations in \aref{sec:MITLvsMITLzi}.

\subsection{\MITL\ vs. \MITLzi}\label{sec:MITLvsMITLzi}

The algorithm (from~\cite{96-MITL}) for constructing a \TAls\ for a
\MITL\ formula $\varphi$ applies a series of syntactic transformations
to $\varphi$ such that the resulting formula 
\begin{inparaenum}
\item is in negated normal form,
\item has at most linearly more distinct subformulas,
\item has the same maximum constant as the original formula, and most
  importantly,
\item is in the normal form given in \aref{def:nf}.
\end{inparaenum}
These transformations can be carried out in polynomial time, and the
construction of the \TAls\ assumes that the \MITL\ formula is in the
normal form given by \aref{def:nf} below.
\begin{definition}[Normal Form~\mbox{\cite[\Def~4.1]{96-MITL}}]\label{def:nf}
\MITL\ formula $\varphi$ is said to be in {\em normal form} \iFF\ it is built from propositions and negated propositions using conjunction,
disjunction, and temporal subformulas of the following six types:
\vfil\penalty999\vfilneg
\begin{multicols}{2}
\begin{myenums}
\item\label{typ:nfFinally} $\ltlF_\II\varphi'$ with $\IsLOpen{\II}$, $\iInf{\II}=0$, and $\iSup{\II}\in\Real$,
\item\label{typ:nfAlways}  $\ltlG_\II\varphi'$ with $\IsLOpen{\II}$, $\iInf{\II}=0$, and $\iSup{\II}\in\Real$,
\item\label{typ:nfAllTime} $\ltlG\varphi'$.
\item\label{typ:nfUntil}   $\Until  {\varphi_1}{\varphi_2}[\II]$ with $\iInf{\II}>0$, and $\iSup{\II}\in\Real$,
\item\label{typ:nfRelease} $\Release{\varphi_1}{\varphi_2}[\II]$ with $\iInf{\II}>0$, and $\iSup{\II}\in\Real$,
\item\label{typ:nfUntil2}  $\Until{\varphi_1}{\varphi_2}$,
\end{myenums}
\end{multicols}
\end{definition}

The main challenge (in terms of complexity) is in handling formulas of
\aref{typ:nfUntil} and \aref{typ:nfRelease}. If the formula you start
with is in {\MITLzi} then it can be seen that the normal form does not
have any subformulas of \aref{typ:nfUntil} and
\aref{typ:nfRelease}. Hence, the \TAls\ constructed for
\MITLzi\ formulas is ``small'', which results in \PSPACE\ decision
procedures.

To see the difficulty of transforming \aref{typ:nfUntil} and
\aref{typ:nfRelease} formulas into \TAlp, consider
$\ltlG_{\Opened{3,4}}\Paren{p\rightarrow\ltlF_{\Opened{1,2}}q}$.
Intuitively, the formula says, during the $4$\textsuperscript{th} unit
of time, every $p$ should be followed by a $q$ within $1$ to $2$ units
of time.  A na\"ive approach, starts and dedicates a clock after
seeing every $p$ during $4$\textsuperscript{th} unit of time, and uses
that clock to ensure that there will be at least one $q$, $1$ to $2$
units of time after the corresponding $p$ was seen.  However, this
approach does not work, since there is no bound on number of $p$ that
one can expect to see during any period of time, which makes number of
required clocks infinite.

Instead, the construction in~\cite{96-MITL} divides $\nnReal$ into
$\LclRop{0,1},\LclRop{1,2},\ldots$ intervals. Two important facts are
central to the construction.
\begin{inparaenum}
\item For any interval $\LclRop{n,n+1}$ and any \aref{typ:nfUntil} $\UntilOp$ or \aref{typ:nfRelease} $\ReleaseOp$ formula $\varphi$, 
			the subset of times in $\LclRop{n,n+1}$ for which $\varphi$ is true is exactly union of two possibly empty intervals.
			Using this property, for each interval $\LclRop{n,n+1}$, we first guess those two intervals and then use at most $4$ clocks to verify our guess.
\item We can start reusing a clock at most $\iSup{\II}$ units of time after we started using it.
			Therefore, total number of clocks required for checking each \aref{typ:nfUntil} and \aref{typ:nfRelease} formula is bounded by $4\iSup{\II}$.
\end{inparaenum}
Since $\iSup{\II}$ could be exponentially big, the resulting \TAls\ could
have exponentially many clocks~\footnote{The construction
  in~\cite{96-MITL}, keeps track of the subset of clocks that are free
  (\ie\ can be reused) in the discrete modes of the \TAls. This makes
  the number of locations doubly exponential. However, it is possible
  to reuse clocks in a queue like fashion and instead of encoding a
  subset of free clocks in discrete modes, one can just encode the
  index of the next free clock.  This approach exponentially decreases
  number of required discrete modes. This optimization however does
  not change the asymptotic complexity.}.

\subsection{Witness Points and Intervals}\label{sec:witness}

Let us define \emph{step size} of an interval $\II$ as follows.
\begin{equation}\label{eq:step}
\Step[\II]\defEQ\left\{
\begin{array}{@{}l@{\ \ \ }l@{}}
        \iSup{\II}-\iInf{\II} & \mbox{if }\iSup{\II}-\iInf{\II}<\iInf{\II} \\
        \iInf{\II}            & \mbox{otherwise}
\end{array}
\right.
\end{equation}
The crucial observation needed to prove \aref{thm:MITLwi} is that the
truth of \aref{typ:nfUntil} and \aref{typ:nfRelease} does not change
very frequently. We show that for a bounded non-empty interval $\II$
with $\iInf{\II} > 0$, using constantly many clocks, the \TAls\ can
monitor the truth of a formula of \aref{typ:nfUntil} or
\aref{typ:nfRelease} for intervals
$\LclRop{0,c},\LclRop{c,2c},\ldots$, where $c\defEQ\Step[\II]$,
instead of intervals $\LclRop{0,1},\LclRop{1,2},\ldots$ as in the
construction given in~\cite{96-MITL}. This has two important
consequences.
\begin{myenums}
\item If a formula is in \MITLwi, then number of required clocks will
  be at most linear in the size of formula.  For example, verifying
  $\ltlG_{\Closed{n,2n}}\varphi$ requires constant number of clocks,
  as opposed to exponentially many clocks in~\cite{96-MITL}.
\item Consider satisfiability of
  $\varphi\defEQ\ltlG_{\Closed{1,2}}\ltlF_{\Closed{0.01,0.02}}\varphi'$
  formula.  The algorithm in~\cite{96-MITL}, first changes $\varphi$
  to an ``equivalent'' formula
  $\varphi\defEQ\ltlG_{\Closed{100,200}}\ltlF_{\Closed{1,2}}\varphi'$,
  because if observation intervals are
  $\LclRop{0,1},\LclRop{1,2},\ldots$ then all constants in the input
  formula must be natural numbers.  Therefore, \TAls\ will have
  hundreds of clocks.  However, we show there is no need for
  observation intervals to have natural numbers as endpoints. This
  means that the \TAls\ for $\varphi$ requires at most $8$ clocks for each
  of $\ltlG_{\Closed{1,2}}$ and $\ltlF_{\Closed{0.01,0.02}}$
  sub-formulas. Note that the algorithm to check emptiness of
  \TAlp\ will replace all rational numbers by natural numbers by
  scaling, when constructing the region graph~\cite{94-TA}. However,
  in spite of this, it is worth observing that the complexity of
  emptiness checking of \TAlp\ has an exponentially worse dependence
  on the number of clocks, than on
  constants~\cite[\Lem~4.5]{94-TA}. Thus, our observations may lead to
  better running times in practice even for {\MITLzi}.
\end{myenums}

\subsubsection{Witness Points for \UntilOp\ Operators.}\label{sec:Uwitness}

For the rest of this section, let us fix an arbitrary signal
$\Sgn$. We begin by presenting some technical definitions of
``witnesses'' that demonstrate when an $\UntilOp$-formula is
satisfied.
\begin{definition}[Witness Sets for \UntilOp]\label{def:until-witness}
For every \MTL\ formulas $\varphi_1$ and $\varphi_2$, and
$i\oftype\Brace{1,2,3}$, we define
$\Uwitness^i\Paren{\varphi_1,\varphi_2}$ to be a subset of $\nnReal^2$
defined by the following predicates over $\Paren{r,w}$:
\begin{myenums}
\item $r\leq w \wedge \Paren*{\forall
  t\oftype\Paren{r,w}\WeHave\Sgn^t\NSat\varphi_1}\wedge\Sgn^w\NSat\varphi_2$
\item $\begin{array}[t]{@{}l@{}l@{}}
  r<w \wedge \exists\epsilon\oftype\pReal\SuchThat
  \Paren*{\forall t\oftype\Paren{r,w}\WeHave\Sgn^t\NSat\varphi_1} & \wedge
  \Paren*{\forall t\oftype\Paren{w\Minos\epsilon ,w}\WeHave\Sgn^t\NSat\varphi_2}
			\end{array}$
\item $\begin{array}[t]{@{}l@{}l@{}}
  r<w \wedge \exists\epsilon\oftype\pReal\SuchThat
  \Paren*{\forall t\oftype\Paren{\makebox[0pt][l]{$r,w+\epsilon$}\phantom{r,w+\epsilon}}\WeHave\Sgn^t\NSat\varphi_1} & \wedge
  \Paren*{\forall t\oftype\Paren{w,w+\epsilon }\WeHave\Sgn^t\NSat\varphi_2}
  \end{array}$
\end{myenums}
\end{definition}
Notice, that if $(r,w)$ is in any of the witness sets given in
\aref{def:until-witness}, then it provides proof that certain until
formulas are true. This is captured by the definition of proof sets,
given next.
\begin{definition}[Proof Sets for \UntilOp]\label{def:until-proofset}
For every \MTL\ formulas $\varphi_1$ and $\varphi_2$, times
$r,w\oftype\nnReal$, interval $\II\oftype\nnII$, and
$i\oftype\Brace{1,2,3}$, we define
$\Uproofset^i\Paren{\varphi_1,\varphi_2,r,w,\II}$ to be a subset of
$\nnReal$ defined by the following predicates over $t$:
\begin{myenums}
\item $\Uwitness^1\Paren{r,w}\wedge r\leq t\wedge w-t\in\II$
\item $\Uwitness^2\Paren{r,w}\wedge r\leq t\wedge w-t\in\OpenLCloseR{\II}$
\item $\Uwitness^3\Paren{r,w}\wedge r\leq t\wedge w-t\in\OpenRCloseL{\II}$
\end{myenums}
\end{definition}
A proof set $\Uproofset^i\Paren{\varphi_1,\varphi_2,r,w,\II}$
establishes the fact that $\Until{\varphi_1}{\varphi_2}[\II]$ is true
at time $r$ in signal $\Sgn$. This is proved next.
\newcommand{\RESUProofSet}{%
For any \MTL\ formulas $\varphi_1$ and $\varphi_2$,
				times $r,w\oftype\nnReal$,
				interval $\II\oftype\nnII$, 
				$i\oftype\Brace{1,2,3}$, and
				$t\oftype\Uproofset^i\Paren{\varphi_1,\varphi_2,r,w,\II}$
we have $\Sgn^t\NSat\Until{\varphi_1}{\varphi_2}[\II]$.}
\PVSaddress{\PVSLabel{mtl\_witness}{def\_until\_witness\_\{1,2,3\}\_proofset}}
\begin{pvsproposition}[Proof Sets for \UntilOp]\label{prop:until-proofset}
\RESUProofSet
\end{pvsproposition}

In \aref{prop:until-proofset}, the signal $\Sgn$ need not be finitely
variable. Also, the formulas $\varphi_1,\varphi_2$ could be any
\MTL\ formulas. The truth of $\Until{\varphi_1}{\varphi_2}[\II]$
within $\LclRop{0,\Step[\II]}$ changes only finitely many times. This
crucial observation helps limit the number of clocks needed to monitor
the truth of $\UntilOp$-subformulas.
\newcommand{\RESUfvar}{%
For any \MTL\ formulas $\varphi_1$ and $\varphi_2$, and interval
$\II\oftype\Brace{\II\oftype\allII\filter\iInf{\II},\iSup{\II}\in\pReal}$,
there are two intervals $T_1\oftype\nnII$ and $T_2\oftype\nnII$ with
the following properties:
\begin{myitems}
\item $\forall t_1\oftype T_1,t_2\oftype T_2\WeHave t_1<t_2$, and
\item $\forall t\oftype\nnReal\WeHave \Paren*{t<\Step[\II] \wedge
  \Sgn^t\NSat\Until{\varphi_1}{\varphi_2}[\II]}
  \Leftrightarrow\Paren{t\in T_1\cup T_2}$
\end{myitems}}
\PVSaddress{\PVSLabel{mtl\_witness}{until\_witness\_interval\_2}}
\begin{pvstheorem}[Finite Variability of \UntilOp]\label{thm:until-fv}
\RESUfvar
\end{pvstheorem}

It is worth noting that \aref{thm:until-fv} is not restricted to
\MITL\ or to finite variable signals. Since within
$\LclRop{0,\Step[\II]}$, the times when
$\Until{\varphi_1}{\varphi_2}[\II]$ is true can be partitioned into two
intervals, suggests that a \TAls\ checking this property can just
guess these intervals. But how can such intervals be guessed?
\aref{def:until-witness} provides an answer. These observations are
combined in the next theorem, to identify what the \TAls\ needs to
guess and check for $\UntilOp$-formulas.
\newcommand{\RESUwitness}{%
In \aref{thm:until-fv}, if $\fvar[\Sgn]$ then $T_1$ and $T_2$ have the following properties:
\begin{myitems}
\item If $T_1\neq\emptyset$ then there are $w_1\oftype\pReal$ and $i\oftype\Brace{1,2}$ such that:
			\begin{myenums}
			\item if $i=1$ then $w_1-\iInf{T_1}\in\II$, otherwise, $w_1-\iInf{T_1}\in\CloseR{\II}$
			\item $\Paren{\iInf{T_1},w_1}\in\Uwitness^i\Paren{\varphi_1,\varphi_2}$
			\item $T_1\subseteq\Uproofset^i\Paren{\varphi_1,\varphi_2,\iInf{T_1},w_1}$
			\end{myenums}
\item If $T_2\neq\emptyset$ then there are $w_2\oftype\pReal$ and $i\oftype\Brace{1,3}$ such that:
			\begin{myenums}
			\item $w_2-\iInf{T_2}\in\CloseR{\II}$
			\item $\Paren{\iInf{T_2},w_2}\in\Uwitness^i\Paren{\varphi_1,\varphi_2}$
			\item $T_2\subseteq\Uproofset^i\Paren{\varphi_1,\varphi_2,\iInf{T_2},w_2}$
			\end{myenums}
\end{myitems}}
\PVSaddress{\PVSLabel{mtl\_witness}{until\_witness\_interval\_3}}
\begin{pvstheorem}[Witness Point for \UntilOp]\label{thm:until-witness-point}
\RESUwitness
\end{pvstheorem}

In \aref{thm:until-witness-point}, the
	1\textsuperscript{st} property bounds possible values of $w_i$, and hence bounds possible values that should be guessed by \TAls.
	The 2\textsuperscript{nd} property specifies what $w_i$ should satisfy (\ie\ what \TAls\ should verify about the guess), and
	the 3\textsuperscript{rd} property states that $w_i$ is enough for proving that $\Until{\varphi_1}{\varphi_2}[\II]$ is satisfied by $\Sgn$ at all times in
	$T_i$.

\subsubsection{Witness Intervals for \ReleaseOp\ Operators.}\label{sec:Rwitness}
We now identify how a timed automaton can check
$\ReleaseOp$-formulas. We will repeat the steps from the previous
section. We will identify witness intervals, and proof sets for
$\ReleaseOp$-formulas. As in the case of $\UntilOp$, these provide
proofs of when a $\ReleaseOp$ formula is true.
\begin{definition}[Witness Interval for \ReleaseOp]\label{def:release-witness}
For every \MTL\ formulas $\varphi_1$ and $\varphi_2$, and 
					$i\oftype\Brace{1,\ldots,4}$, 
we define $\Rwitness^i\Paren{\varphi_1,\varphi_2}$ to be a subset of $\IIc$ defined by the following predicates over \II:
\[
\setcounter{enumi}{0}
\newcommand{\Bullet}{\stepcounter{enumi}\mbox{\labelenumi}}
\begin{array}{@{}l@{\hspace{\labelsep}}l@{}l@{}l@{}}
\Bullet&\multicolumn{3}{@{}l@{}}{\forall t\oftype\II\WeHave\Sgn^t\NSat\varphi_2}\\
\Bullet&				\II &\neq\emptyset\wedge\IsLClosed{\II}	&\wedge\forall t\oftype\II\WeHave\Sgn^t\NSat\varphi_1\\
\Bullet&				\II &\neq\emptyset\wedge\IsRClosed{\II}	&\wedge\forall t\oftype\II\WeHave\Sgn^t\NSat\varphi_2\wedge\Sgn^{\iSup{\II}}\NSat\varphi_1\\
\Bullet&\CloseL{\II}&\neq\emptyset\wedge\IsRClosed{\II}	&
			\wedge\forall t\oftype\II\WeHave\Sgn^t\NSat\varphi_2\wedge\exists\epsilon\oftype\pReal\SuchThat
			 \forall t\oftype\Paren{\iSup{\II},\iSup{\II}+\epsilon}\WeHave
			 \Sgn^t\NSat\varphi_1
\end{array}	
\]
\end{definition}

\begin{definition}[Proof Sets for \ReleaseOp]\label{def:release-proofset}
For every \MTL\ formulas $\varphi_1$ and $\varphi_2$,
          intervals $\II,\JJ\oftype\nnII$, and
          $i\oftype\Brace{1,2,3,4}$, 
we define $\Uproofset^i\Paren{\varphi_1,\varphi_2,\II,\JJ}$ to be a subset of $\nnReal$ defined by the following predicates over $t\oftype\nnReal$:
\[
\setcounter{enumi}{0}
\newcommand{\Bullet}{\stepcounter{enumi}\mbox{\labelenumi}}
\begin{array}{@{}l@{\hspace{\labelsep}}l@{}l@{}l@{}}
\Bullet&	\II\in\Rwitness^1\Paren{\varphi_1,\varphi_2}\wedge&	t+\JJ\subseteq\II\\
\Bullet&	\II\in\Rwitness^2\Paren{\varphi_1,\varphi_2}\wedge& t+\JJ\subseteq\Paren{\iInf{\II},\infty} &\wedge	t<   \iSup{\II}\wedge\iInf{\JJ}>0\\
\Bullet&	\II\in\Rwitness^3\Paren{\varphi_1,\varphi_2}\wedge&	t+\JJ\subseteq\II+\nnReal								&\wedge	t<   \iSup{\II}\\
\Bullet&	\II\in\Rwitness^4\Paren{\varphi_1,\varphi_2}\wedge& t+\JJ\subseteq\II+\nnReal								&\wedge	t\leq\iSup{\II}
\end{array}
\]
\end{definition}

\PVSaddress{\PVSLabel{mtl\_witness}{def\_release\_witness\_\{1,2,3,4\}\_proofset}}
\begin{pvsproposition}[Proof Sets for \ReleaseOp]\label{prop:release-proofset}
For any \MTL\ formulas $\varphi_1$ and $\varphi_2$,
        intervals $\II,\JJ\oftype\nnII$,
        $i\oftype\Brace{1,2,3,4}$, and
				$t\oftype\Rproofset^i\Paren{\varphi_1,\varphi_2,\II,\JJ}$
we have $\Sgn^t\NSat\Release{\varphi_1}{\varphi_2}[\JJ]$.				
\end{pvsproposition}

Like $\UntilOp$-formulas, a formula
$\Release{\varphi_1}{\varphi_2}[\II]$ changes its truth only finitely
many times in the interval $\LclRop{0,\Step[\II]}$.
\PVSaddress{\label{ft:release}\PVSLabel{mtl\_witness}{release\_witness\_interval\_1}}
\begin{pvstheorem}[Finite Variability of \ReleaseOp]\label{thm:release-fv}
For any \MTL\ formulas $\varphi_1$ and $\varphi_2$, and
        non-empty positive bounded interval \II, 
there are four intervals $T_1,\ldots,T_4$ with the following properties:
\setlength{\labelsep}{0.pt}
\begin{myitems}
\item $\forall i,j\oftype\Brace{1,\ldots,4}, t_i\oftype T_i,t_j\oftype T_j\WeHave i<j\Rightarrow t_i<t_j$, and
\item $\forall t\oftype\nnReal\WeHave \Paren*{t<\Step[\II]\wedge\Sgn^t\NSat\Release{\varphi_1}{\varphi_2}[\II]}\Leftrightarrow
	\Paren{t\in\bigcup_{i\oftype{1,\ldots,4}} T_i}$
\end{myitems}
\end{pvstheorem}

Like in \aref{thm:until-witness-point}, we can combine
\aref{thm:release-fv} and \aref{def:release-witness} to come up with
how a \TAls\ can check such $\ReleaseOp$ formulas.
\PVSaddress*{ft:release}
\begin{pvstheorem}[Witness Interval for \ReleaseOp]\label{thm:release-witness-point}
In \aref{thm:release-fv}, if $\fvar[\Sgn]$ then $T_1,\ldots,T_4$ have the following properties:
\begin{multicols}{2}
\begin{myitems}
\item If $T_1\neq\emptyset$ then $\exists\II\oftype\IIc$ such that:
			\begin{myenums}
			\item $\II\subseteq\Paren{0,\Step[\JJ]}$
			\item $\Rwitness^2\Paren{\varphi_1,\varphi_2,\II}$
			\item $T_1\subseteq\Rproofset^2\Paren{\varphi_1,\varphi_2,\II,\JJ}$
			\end{myenums}
\item If $T_2\neq\emptyset$ then $\exists\II\oftype\IIc$ such that:
			\begin{myenums}
			\item $\II\subseteq\LclRop{\Step[\JJ],\Step[\JJ]+\iSup{\JJ}}$
			\item $\Rwitness^1\Paren{\varphi_1,\varphi_2,\II}$
			\item $T_2\subseteq\Rproofset^1\Paren{\varphi_1,\varphi_2,\II,\JJ}$\\
			\end{myenums}
\item If $T_3\neq\emptyset$ then $\exists\II\oftype\IIc$ such that:
			\begin{myenums}
			\item $\II\subseteq\LclRop{\Step[\JJ],\Step[\JJ]+\iSup{\JJ}}$
			\item $\Rwitness^4\Paren{\varphi_1,\varphi_2,\II}$
			\item $T_3\subseteq\Rproofset^4\Paren{\varphi_1,\varphi_2,\II,\JJ}$
			\end{myenums}
\item If $T_4\neq\emptyset$ then $\exists\II\oftype\IIc,i\oftype\Brace{2,3,4}$ such that:
			\begin{myenums}
			\item $\II\subseteq\LclRop{\Step[\JJ],\Step[\JJ]+\iSup{\JJ}}$
			\item $\Rwitness^i\Paren{\varphi_1,\varphi_2,\II}$
			\item $T_4\subseteq\Rproofset^i\Paren{\varphi_1,\varphi_2,\II,\JJ}$
			\end{myenums}
\end{myitems}
\end{multicols}
\end{pvstheorem}

\subsubsection{Constructing a \TAls\ for \MITLwi.}
One can use \aref{thm:until-witness-point} and
\aref{thm:release-witness-point} and follow the same ideas outlined
in~\cite{96-MITL} to construct a \TAls\ for \aref{typ:nfUntil} and
\aref{typ:nfRelease} formulas.  Let us outline how this works for
\aref{typ:nfUntil} formulas. If the automaton guesses that $T_1$ is not
empty then it must make this guess at time exactly $\iInf{T_1}$. At
the same time, the automaton takes two more actions: First, it resets
a {\em free} clock $x$ and remembers that this clock is not free
anymore.  Second, it guesses whether $i$ should be $1$ or
$2$. Suppose, $i$ is chosen to be $1$.  As long as $x$ is not free, the
automaton makes sure that the input signal satisfies $\varphi_1$.
Note that this is a different proof obligation and will be considered
by an induction on the structure of input formula.
At the same time or at some time later, the automaton should guess
whether the current time is $\iSup{T_1}$.  At any point in time, if
the automaton does not make that call (\ie\ decides the current time
is not $\iSup{T_1}$), it means the automaton wants to prove that the
input signal satisfies the \UntilOp\ formula at all points in time
between $\iInf{T_1}$ and sometime in the future.
As soon as the automaton guesses that the current time is
$\iSup{T_1}$, it resets a new free clock $y$ and marks it as non-free.
The automaton then makes sure that when current values of $x$ and $y$
belong to $\II$, the input signal satisfies $\varphi_2$ at least once.
As soon as $\varphi_2$ becomes true during during this period, the
proof obligation is over and $x$ and $y$ will both be marked as free
clocks (note that $\varphi_1$ does not need to be true when
$\varphi_2$ becomes true).
Using \aref{thm:until-witness-point}, we know what the automaton
checks, guarantees $\forall t\oftype T_1\WeHave
f^t\NSat\Until{\varphi_1}{\varphi_2}[\II]$.  However, the automaton
has only finitely many clocks and it cannot reuse a clock while it is
not free.  The significance of \aref{thm:until-fv} is that it
guarantees simultaneously guessing and proving at most
$\Ceil{\frac{\iInf{\II}}{\Step[\II]}}+1$ number of $T_1$, $T_2$
intervals is enough. Since clocks $x$ and $y$ will be freed at most
$\iSup{\II}$ units of time after they became non-free, number of
required clocks for each \aref{typ:nfUntil} formula will be only twice
the number of simultaneous proof obligations.  The same argument holds
for \ReleaseOp\ operators, except that the automaton has to
simultaneously guess and prove at most
$\Ceil{\frac{\iInf{\II}}{\Step[\II]}}+1$ number of $T_1$, $T_2$,
$T_3$, $T_4$ intervals.

\section{Conclusion}\label{sec:conclusion}

We proved that the classical decision procedures for satisfiability
and model checking of {\MITL}~\cite{96-MITL} are incorrect. This is
because they rely on a semantics for the $\ReleaseOp$ operator which
is not the dual of $\UntilOp$. We give a new semantics of $\ReleaseOp$
and prove that it behaves like the dual of $\UntilOp$ over signals
that are finitely variable. Identifying the right semantics for
$\ReleaseOp$ is subtle as we show that it is not possible to give a
correct semantics using characterization that uses only two quantified
variables. Using the new semantics, we give a translation from {\MITL}
to {\TAlp} and thereby correcting the decision procedures for
{\MITL}. Finally, we also identify a fragment of {\MITL} called
{\MITLwi}, that is more expressive than {\MITLzi}, but nonetheless has
decision procedures in {\PSPACE}.

\bibliographystyle{plain}
\bibliography{biblio}

\appendix
\section{Proofs}

Most of our positive results have been already proven in \PVS\ and references to those proofs have already been given in the main part of this paper.
In this section, we put proof ideas for many of those results to give a flavor of what is done in those proofs.
\aref{prop:subsize} is not proved in \PVS, so we have its full proof here.
Also, \aref{thm:duality} is an immediate consequence of new \aref{lem:duality-1} and new \aref{lem:duality-2}, which we sketch their proofs here and refer the 
reader to \PVS\ for the full proofs.


\begin{replemma}{lem:fvar}
\RESfvar
\begin{proof}[\ProofIdea]
From right to left direction is trivially true.
The other direction can be proved using induction on the structure of $\varphi$.
The only interesting cases are when $\varphi$ is in the form of $\Until{\varphi_1}{\varphi_2}[\II]$ or $\Release{\varphi_1}{\varphi_2}[\II]$.
Either way, use induction hypothesis and let $\epsilon_1,\ldots,\epsilon_6\oftype\pReal$ be finite variability constants for $\varphi_1$ and
$\varphi_2$ at times $0$, $\iInf{\II}$, and $\iSup{\II}$ (if supremum is finite). Let $\epsilon\oftype\pReal$ be any value that is 
$<\min\Brace{\epsilon_1,\ldots,\epsilon_6}$ and prove the new satisfiability relation does not change its value during $\Paren{0,\epsilon}$.
\end{proof}
\end{replemma}


\PVSaddress{\PVSLabel{mtl}{sat\_implies\_nnfsat}}
\newcommand{\RESDualityOne}{%
  For any signal $\Sgn$ that is finitely variable from right and \MTL\ formula $\varphi$, we have $\Sgn\NSat\varphi$ implies $\Sgn\NSat\nnf[\varphi]$.}
\begin{pvslemma}[Duality-1]\label{lem:duality-1}
\RESDualityOne
%
\begin{proof}[\ProofSketch]
Proof is by induction on the structure of $\varphi$.
We prove the case $\varphi\defEQ\neg\Paren{\Release{\varphi_1}{\varphi_2}[\II]}$ here.
We know all the following conditions are true:
\[
\begin{array}{@{}l@{}}
  \exists t_1\oftype\II\SuchThat\Sgn^{t_1}\NSat\neg\varphi_2  
  \\
  \forall t_1\oftype\pReal\WeHave\Paren*{\Sgn^{t_1}\NSat\varphi_1}\Rightarrow
  \exists t_2\oftype\Braket{0,t_1}\cap\II\SuchThat\Sgn^{t_2}\NSat\neg\varphi_2
  \\
  \forall t_1\oftype\CloseL{\II},t_2\oftype\II\cap\Paren{t_1,\infty}\WeHave\exists t_3\oftype\II\SuchThat
    \Paren{t_3\leq t_1\wedge\Sgn^{t_3}\NSat\neg\varphi_2}\vee\\
    \multicolumn{1}{@{}r@{}}{\Paren{t_1 < t_3\leq t_2\wedge\Sgn^{t_3}\NSat\neg\varphi_1}}
\end{array}
\]
Let $A\defEQ\Brace{t\oftype\II\filter\Sgn^t\NSat\neg\varphi_2}$ and $t_1\defEQ\inf A$.
Pick any $t\oftype\pReal$ and assume $\Sgn^t\NSat\varphi_1$ is true
(if no such $t$ exists then the first condition immediately give us $\Sgn\NSat\Until{\neg\varphi_1}{\neg\varphi_2}[\II]$).
Because of the second condition, there must be a $t'\oftype\Braket{0,t}\cap\II$ such that $\Sgn^{t'}\NSat\neg\varphi_2$.
Therefore, $t_1\leq t'$ and hence $t_1\leq t$.
This means $\forall t_2\oftype\Paren{0,t_1}\WeHave\Sgn^{t_2}\NSat\neg\varphi_1$.
\begin{myitems}
\item If $\Paren{\Sgn^{t_1}\NSat\neg\varphi_2}\wedge t_1\in\II$ then $\Sgn\NSat\Until{\neg\varphi_1}{\neg\varphi_2}[\II]$ is true.
\item If $\Paren{\Sgn^{t_1}\NSat    \varphi_2}\wedge t_1\in\II$ then 
      fix $t_1$ in the third condition and, knowing $\II\cap\Paren{t_1,\infty}\neq\emptyset$, let $t_2\oftype\II\cap\Paren{t_1,\infty}$ be an 
      arbitrary element.
      For some $t_3\oftype\II$ we have
      $\Paren{t_3\leq t_1\wedge\Sgn^{t_3}\NSat\neg\varphi_2}\vee\Paren{t_1 < t_3\leq t_2\wedge\Sgn^{t_3}\NSat\neg\varphi_1}$.
      Since $t_1=\inf A$  the left disjunct is false. Therefore, we must have 
      $\Paren{t_1 < t_3\leq t_2\wedge\Sgn^{t_3}\NSat\neg\varphi_1}$.
      Since we can make $t_2$ arbitrary close to $t_1$, using $\fvarR[\Sgn,\neg\varphi_1]$, \wLOG, we assume $t_2\oftype\II\cap\Paren{t_1,\infty}$ 
      is such that $\forall t_3\oftype\LopRcl{t_1,t_2}\WeHave\Sgn^{t_3}\NSat\neg\varphi_1$.
      \begin{myitems}
      \item If $\Sgn^{t_1}\NSat\neg\varphi_1$ then 
            $\forall t'\oftype\LopRcl{0,t_2}\WeHave\Sgn^{t'}\NSat\neg\varphi_1$, and since $t_1=\inf A$, we know 
            $\exists t\oftype\Braket{0,t_2}\cap\II\SuchThat\Sgn^{t}\NSat\neg\varphi_2$.
            Therefore, $\Sgn\NSat\Until{\neg\varphi_1}{\neg\varphi_2}[\II]$.
      \item If $\Sgn^{t_1}\NSat\varphi_1$ then because of the second condition there is $t''\oftype\Braket{0,t_1}\cap\II$ such that 
            $\Sgn^{t''}\NSat\neg\varphi_2$. But this is contradictory to the facts $t_1=\inf A$ and $\Sgn^{t_1}\NSat\varphi_2$.
      \end{myitems}
\item If $t_1\not\in\II$ then we know $t_1\in\CloseL{\II}$.
      Follow the steps of the previous case to obtain the same contradiction.
\end{myitems}
\end{proof}
\end{pvslemma}

Looking at the proof of \Lem~\ref{lem:duality-1}, we see that for temporal operators 
$\Until{\varphi_1}{\varphi_2}[\II]$ and $\Release{\varphi_1}{\varphi_2}[\II]$, only right-side finite variability of $\varphi_1$ is used.
Therefore, as far as \Lem~\ref{lem:duality-1} and hence \Thm~\ref{thm:duality} are concerned, right-side finite variability of atomic propositions 
that are only appeared in the right hand side of temporal operators is not required.


\PVSaddress{\PVSLabel{mtl}{nnfsat\_implies\_sat}}
\newcommand{\RESDualityTwo}{%
  For any signal $\Sgn$ and \MTL\ formula $\varphi$, we have $\Sgn\NSat\nnf[\varphi]$ implies $\Sgn\NSat\varphi$.}
\begin{pvslemma}[Duality-2]\label{lem:duality-2}
\RESDualityTwo
%
\begin{proof}[\ProofSketch]
Proof is by induction on the structure of $\varphi$.
We prove the case $\varphi\defEQ\neg\Paren{\Release{\varphi_1}{\varphi_2}[\II]}$ here.
If $\II=\emptyset$ then $\Sgn\NNSat\nnf[\varphi]$, therefore for the rest of this proof assume $\II\neq\emptyset$.
Let $t_1\oftype\II$ be any element that satisfies 
  $\Paren{\Sgn^{t_1}\NSat\neg\varphi_2}\wedge\forall t_2\oftype\Paren{0,t_1}\WeHave\Paren{\Sgn^{t_2}\NSat\neg\varphi_1}$.
We consider every case in \Def~\ref{def:mtl-newsem} for $\Sgn\NSat\Release{\varphi_1}{\varphi_2}[\II]$ and show 
they are all lead to contradiction.
\begin{myitems}
\item If $\forall t\oftype\II\WeHave\Sgn^t\NSat\varphi_2$ then 
      we have a contradiction since $t_1\in\II\wedge\Paren{\Sgn^{t_1}\NSat\neg\varphi_2}$.
\item If for some $t\oftype\pReal$ we have 
      $\Paren*{\Sgn^t\NSat\varphi_1}\wedge\forall t'\oftype\Braket{0,t}\cap\II\WeHave\Sgn^{t'}\NSat\varphi_2$ then 
      $t_1\leq t$, since $\forall t_2\oftype\Paren{0,t_1}\WeHave\Paren{\Sgn^{t_2}\NSat\neg\varphi_1}$.
      Knowing $t_1\in\II$ and $\forall t'\oftype\Braket{0,t}\cap\II\WeHave\Sgn^{t'}\NSat\varphi_2$, we reach to a contradiction
      $\Sgn^{t_1}\NSat\neg\varphi_2$.
\item If for some $t\oftype\CloseL{\II}$ and $t'\oftype\II\cap\Paren{t,\infty}$ we have
      $\forall t''\oftype\II\WeHave\Paren{t''\leq t\Implies\Sgn^{t''}\NSat\varphi_2}\wedge\Paren{t<t''\leq t'\Implies\Sgn^{t''}\NSat\varphi_1}$
      then 
      from  $t_1\in\II$, 
            $\Sgn^{t_1}\NSat\neg\varphi_2$, and 
            $\forall t''\oftype\II\cap\Braket{0,t}\WeHave\Sgn^{t''}\NSat\varphi_2$ 
      we conclude $t<t_1$.
      Knowing $t<t'$ and by setting $t''\defEQ\frac{1}{2}\Paren{t+\min\Brace{t_1,t'}}$ we have
      $t<t''<\min\Brace{t_1,t'}$ and hence $t''\in\II$.
      Therefore, $\Sgn^{t''}\NSat\varphi_1$ which is contrary to the assumption 
      $\forall t_2\oftype\Paren{0,t_1}\WeHave\Paren{\Sgn^{t_2}\NSat\neg\varphi_1}$.
\end{myitems}
\end{proof}
\end{pvslemma}


\begin{reptheorem}{thm:duality}
  \RESnnfEqvSat
\begin{proof}
Immediate from \Lem~\ref{lem:duality-1} and \Lem~\ref{lem:duality-2}.
\end{proof}
\end{reptheorem}


\begin{repproposition}{prop:next}
\RESnext
\begin{proof}[\ProofIdea]
Apply two easy inductions on the structure of $\varphi$ 
(one for \Def~\ref{def:mtl-sem} and one for \Def~\ref{def:mtl-newsem}).
\end{proof}
\end{repproposition}


\begin{repproposition}{prop:subsize}
\PROPsubsize
\begin{proof}
During this proof, $\varphi'$, $\varphi'_1$ and $\varphi'_2$ stand for $\toOld[\varphi]$, $\toOld[\varphi_1]$, and $\toOld[\varphi_2]$, respectively.
Define $\Ops$ as a function that maps any \MTL\ formula $\varphi$ to the set of \MTL\ formulas that are operands of $\varphi$. Also, 
define $\Height$ as a function that maps $\top$, $\bot$, and every atomic proposition to $1$, and maps every other \MTL\ formula to $1$ plus height of its 
highest operand.
Let $\AAA$ be an arbitrary finite non-empty set of \MTL\ formulas, and define $\Height[\AAA]\defEQ\max_{\varphi\oftype\AAA}\Height[\varphi]$ to be the height 
of a highest formula in $\AAA$.

We use induction on $\Height[\AAA]$ to prove $\Size{\bigcup_{\varphi\oftype\AAA}\Sub_{\varphi'}}\leq6\Size{\bigcup_{\varphi\oftype\AAA}\Sub_\varphi}$.
Base of induction is where $\AAA\subseteq\Brace{\top,\bot}\cup\AP$, which is trivially true.
For the inductive step, let 
  $\BBB\defEQ\Brace{\varphi\oftype\AAA\filter\Height[\AAA]=\Height[\varphi]}$ be the set of all formulas in $\AAA$ that have the same height as $\AAA$. Also, 
let $\BBB_1\subseteq\BBB$ be the set of formulas in $\BBB$ that are {\em not} of the form $\Release{\varphi_1}{\varphi_2}[\II]$, and 
let $\BBB_2\defEQ\BBB\setminus\BBB_1$ be everything else in $\BBB$.
We know 
\[
  \bigcup_{\varphi\oftype\AAA}             \Sub_\varphi
  = 
  \bigcup_{\varphi\oftype\AAA\setminus\BBB}\Sub_\varphi \cup 
  \bigcup_{\varphi\oftype\BBB}\Sub_\varphi
  = 
  \overbrace{
  \bigcup_{\varphi\oftype\AAA\setminus\BBB}\Sub_\varphi \cup
  \bigcup_{\substack{\psi\oftype\BBB\\\varphi\oftype\Ops[\psi]}}\Sub_\varphi}^{\CCC} \cup
  \BBB_1\cup
  \BBB_2
\]
\begin{align*}
  \bigcup_{\varphi\oftype\AAA}             \Sub_{\varphi'}
  = & 
  \bigcup_{\varphi\oftype\AAA\setminus\BBB}\Sub_{\varphi'} \cup 
  \bigcup_{\varphi\oftype\BBB}\Sub_{\varphi'}
  \subseteq 
  \overbrace{
    \bigcup_{\varphi\oftype\AAA\setminus\BBB}\Sub_{\varphi'} \cup
    \bigcup_{\substack{\psi\oftype\BBB\\\varphi\oftype\Ops[\psi]}}\Sub_{\varphi'}}^{\CCC'} \cup
  \overbrace{
  \bigcup_{\varphi\oftype\BBB_1}\Brace{\varphi'}}^{\BBB'_1}
  \cup
  \\&
  \overbrace{
  \bigcup_{\Release{\varphi_1}{\varphi_2}[\II]\oftype\BBB_2}\Brace{
      \varphi',
      \underbrace{\Release{\varphi'_1}{\varphi'_2}[\II]}_{\psi_1},
            \underbrace{\Release{\ltlN\varphi'_1}{\varphi'_2}[\II]}_{\psi_2},
              \ltlN\varphi'_1,
              \varphi'_2\wedge\ltlN\varphi'_1,
              \psi_1\vee\psi_2
          }}^{\BBB'_2}
\end{align*}
The inclusion holds, because subformulas of operands of formulas in $\BBB$ are already considered in $\CCC'$.
Using induction hypothesis and knowing 
  $\bigcup_{\varphi\oftype\CCC}\Sub_{\varphi}=\CCC$ and
  $\bigcup_{\varphi\oftype\CCC'}\Sub_{\varphi}=\CCC'$, 
we conclude $\Size{\CCC'}\leq6\Size{\CCC}$.
The proof is complete once we notice, 
  $\CCC$, $\BBB_1$, and $\BBB_2$ are pairwise disjoint,
  $\Size{\BBB'_1}\leq\Size{\BBB_1}$, and 
  $\Size{\BBB'_2}\leq6\Size{\BBB_2}$.
\end{proof}
\end{repproposition}

\begin{wrapfigure}{r}{48mm}%
\raggedleft
\vskip-2\baselineskip
\scalebox{0.64}{\begin{tikzpicture}[>=stealth']
  \definecolor{blue}  {rgb}{0.88,0.88,0.996}
  \definecolor{gray}  {rgb}{0.90,0.90,0.900}
  \definecolor{BGgray}{rgb}{0.97,0.97,0.970}
  \definecolor{red}   {rgb}{0.90,0.50,0.500}

  \tikzset{Fml/.style= {circle ,inner sep=0em,align=center,fill=blue,draw=black,line width=0.2mm,blur shadow={shadow blur steps=5}}}
  \tikzset{Atom/.style={ellipse,inner sep=0em,align=center,fill=gray,draw=black,line width=0.2mm,blur shadow={shadow blur steps=5}}}
  \tikzset{Compose/.style={rectangle,fit=(vee1)(vee2)(wedge)(rel1)(rel2)(nxt), fill=BGgray,draw=black,line width=0.2mm,dashed}}

  \path
  (-4.0 , 0.0 ) node[Fml ,minimum size  =6mm] (rel)   {$\ReleaseOp_\II$}
  (-4.85,-1.5 ) node[Atom,minimum height=6mm,minimum width=10mm] (phi1)  {$\varphi_1$}
  (-3.15,-1.5 ) node[Atom,minimum height=6mm,minimum width=10mm] (phi2)  {$\varphi_2$}
  ;

  \path
  ( 0.0 , 0.0 ) node[Fml ,minimum size  =6mm] (vee1)  {$\vee$}
  (-1.0 ,-1.0 ) node[Fml ,minimum size  =6mm] (vee2)  {$\vee$}
  ( 1.0 ,-1.0 ) node[Fml ,minimum size  =6mm] (wedge) {$\wedge$}
  (-1.5 ,-2.0 ) node[Fml ,minimum size  =6mm] (rel1)  {$\ReleaseOp_\II$}
  (-0.5 ,-2.0 ) node[Fml ,minimum size  =6mm] (rel2)  {$\ReleaseOp_\II$}
  ( 1.0 ,-2.75) node[Fml ,minimum size  =6mm] (nxt)   {$\ltlN$}
  (-0.85,-4.0 ) node[Atom,minimum height=6mm] (op1)   {$\toOld[\varphi_1]$}
  ( 0.85,-4.0 ) node[Atom,minimum height=6mm] (op2)   {$\toOld[\varphi_2]$}
  ;

  \begin{pgfonlayer}{bg} 
  \path
  (-2,0) node[Compose] (old) {}
  ;
  \end{pgfonlayer}

  \path[->,black,thick]
  (rel)  edge (phi1) 
  (rel)  edge (phi2)
  ;
  
  \path[->,black,thick]
  (vee1)  edge (vee2) 
  (vee1)  edge (wedge)
  (vee2)  edge (rel1) 
  (vee2)  edge (rel2) 
  (wedge) edge (nxt)  
  (rel1)  edge (op1)  
  (rel1)  edge (op2)  
  (nxt)   edge (op1)  
  (rel2)  edge (op2)  
  (wedge)           edge[bend left= 45] node [right] {} (op2)
  (rel2.south west) edge[bend left=-45] node [right] {} (nxt.west)
  ;

  \path[->,black,thick,dashed,line width=0.1mm]
  (rel)   edge[bend left= 45] node [right] {} (old.north west)
  ;

  \draw[->,black,thick,dashed,line width=0.1mm] (phi1.south) .. controls ([yshift=-20mm] phi1) and ([xshift=-20mm]               op1) .. (op1.west);
  \draw[->,black,thick,dashed,line width=0.1mm] (phi2.south) .. controls ([yshift=-30mm] phi2) and ([xshift=-30mm, yshift=-15mm] op2) .. (op2.south west);

\end{tikzpicture}  }
\vskip-2\baselineskip
\caption{}\label{fig:fml-dag}
\vskip-3\baselineskip
\end{wrapfigure}
\Fig~\ref{fig:fml-dag} shows the proof idea.
We look at $\varphi$ as a directed acyclic graph, with nodes being subformulas and edges between every node and its operands.
Function $\toOld$ changes this graph. The most interesting case is where $\varphi$ is of the form $\Release{\varphi_1}{\varphi_2}[\II]$.
As it is shown in this figure, $\ReleaseOp$ nodes are replaced by at most $6$ nodes. 
Note that $\vee$ and $\wedge$ operators are associative.
Also, expanding $\ltlN$ to its definition, does not change number of nodes.


\begin{repproposition}{prop:until-proofset}
\RESUProofSet
\begin{proof}[\ProofIdea]
This is a simple application of \UntilOp\ semantics.
The case $i=1$ holds trivially.
For the other two cases, show the extra condition in \aref{def:until-proofset} is enough to find $w'$ such that $(r,w')$ is in the \UntilOp\ proof set of 
type~1.
\end{proof}
\end{repproposition}

\begin{reptheorem}{thm:until-fv}
\RESUfvar
\begin{proof}[\ProofIdea]
Let $T_1$ be the set of all points $t\oftype\LclRop{0,\Step[\II]}$ for which not only $f^t\NSat\Until{\varphi_1}{\varphi_2}[\II]$ is true, but also there is a 
$w\oftype\II$ such that $t\in\Uproofset^1(\varphi_1,\varphi_2,\iInf{T_1},w,\II)$.
Let $T_2$ be the set of all points $t\oftype\LclRop{0,\Step[\II]}$ for which $f^t\NSat\Until{\varphi_1}{\varphi_2}[\II]$ and $t\notin T_1$ are both true.
Use the fact that every point in $T_1$ and $T_2$ has a \UntilOp\ witness of type~1 and show $T_1$ and $T_2$ are both convex sets.
Therefore, we can look at them as intervals. This proves the second condition of the theorem.
To prove the first condition, first note that, by definition, $T_1\cap T_2=\emptyset$.
Use the fact that witness $w$ for any point in $T_1$ belong to $\II$ and $T_1\cap T_2=\emptyset$, to show that every witness $w$ for a point in $T_2$ is not 
only outside of $\II$ but also non-strictly larger than supremum of \II. Conclude every point in $T_1$ is strictly smaller than all points in $T_2$. 
\end{proof}
\end{reptheorem}

\begin{reptheorem}{thm:until-witness-point}
\RESUwitness
\begin{proof}[\ProofIdea]
Continue the proof of \aref{thm:until-fv} and consider the set of witnesses that can be used for points in $T_1$ and the set of witnesses that can be used for 
points in $T_2$. Call these sets $W_1$ and $W_2$.
Define $w_1\defEQ\iSup{W_1}$ and $w_2\defEQ\iInf{W_2}$.
If $w_1\in W_1$ then 
  use \UntilOp\ witness of type~1 (\ie\ $i=1$ in the first part). Otherwise, knowing $f$ is finitely variable from left, 
  use \UntilOp\ witness of type~2 (\ie\ $i=2$ in the first  part). Similarly,
if $w_2\in W_2$ then 
  use \UntilOp\ witness of type~1 (\ie\ $i=1$ in the second part). Otherwise, knowing $f$ is finitely variable from right,
  use \UntilOp\ witness of type~3 (\ie\ $i=3$ in the second part). 
The rest is about proving the choices for $i$ and $w_1$, $w_2$ satisfy all the conditions in this theorem.
\end{proof}
\end{reptheorem}

Proofs of \aref{prop:release-proofset}, \aref{thm:release-fv}, and \aref{thm:release-witness-point} follow the exact same steps as in the corresponding 
results for $\UntilOp$ operator. However, since semantics of \ReleaseOp\ operator, as defined in \aref{def:mtl-newsem}, is more complex than the semantics of 
\UntilOp\ operator, proofs are more involved. For example, instead of only two intervals $T_1$ and $T_2$, we needed four intervals $T_1,\ldots,T_4$ for 
\ReleaseOp\ operator.

\end{document}